\documentclass[12pt]{article}
\usepackage{amsmath}
\usepackage{amsthm}
\theoremstyle{plain}
\newtheorem{assumption}{Assumption}
\usepackage{graphicx}
\usepackage{enumerate}
\usepackage{multirow} 
\usepackage[round,compress,sort]{natbib}
\usepackage{url} 
\usepackage{soul}

\usepackage{amssymb }
\usepackage{enumitem}
\usepackage{algorithm}
\usepackage{algpseudocode}
\usepackage{xcolor}

\newcommand{\RR}{\mathbb{R}}

\newcommand{\X}{\mathcal{X}}

\usepackage{chapterbib}
\usepackage[colorlinks=true, allcolors=blue, hyperfootnotes=false]{hyperref}
\usepackage{xurl} 
\usepackage{titletoc}

\newtheorem{theorem}{Theorem}[section]
\newtheorem{lemma}[theorem]{Lemma}

\newtheorem{proposition}[theorem]{Proposition}
\newtheorem{corollary}[theorem]{Corollary}

\usepackage{setspace}
\usepackage{booktabs} 
\usepackage{caption} 
\usepackage{siunitx} 
\usepackage{threeparttable}
\usepackage{bbm}
\usepackage{xcolor}
\usepackage{threeparttable}
\usepackage{longtable}
\providecommand{\X}{\mathcal{X}}

\providecommand{\RR}{\mathbb{R}}

\newcommand{\anon}{1}

\begin{document}

\date{}
\def\spacingset#1{\renewcommand{\baselinestretch}%
{#1}\small\normalsize} \spacingset{1.}


\if1\anon
{
 \title{\bf Robust Bayesian Optimization via Tempered Posteriors}

\author{
  Jiguang Li\thanks{Email: jiguang@chicagobooth.edu. Jiguang Li is a 4th-year doctoral student in Econometrics and Statistics at the Booth School of Business of the University of Chicago. The author would like to thank Sean O'Hagan and Veronika Rockova for their valuable inputs for the manuscript.}\hspace{.2cm}\\
  Booth School of Business, University of Chicago
  \and \\
  Hengrui Luo\thanks{Corresponding Author. Email: hl180@rice.edu}\hspace{.2cm}\\
  Department of Statistics, Rice University \\
  and \\
  Computational Research Division, Lawrence Berkeley National Laboratory 
}
  \maketitle
} \fi

\if0\anon
{
  \bigskip
  \bigskip
  \bigskip
  \begin{center}
    {\LARGE\bf Robust Bayesian Optimization via Tempered Posteriors}
\end{center}
  \medskip
} \fi

\bigskip
 
\begin{abstract}
Bayesian optimization (BO) iteratively fits a Gaussian process (GP) surrogate to accumulated evaluations and selects new queries via an acquisition function.  Under local misspecification, this feedback loop can produce overconfidence precisely in the region guiding subsequent decisions. We develop a tempered GP-based BO framework that raises the likelihood to a
power $\alpha\in(0,1]$. For a generalized family of improvement acquisitions indexed by $g$, including probability of improvement (PI, $g=0$) and expected improvement (EI, $g=1$), we derive finite-time cumulative regret bounds with adaptively learned kernel hyperparameters. The analysis shows that tempering reduces the noise-driven confidence and information-gain contributions to regret, while a deterministic RKHS term prevents arbitrarily aggressive tempering from being uniformly beneficial. It also clarifies the role of the acquisition function: positive-order $g$-EI rules preserve the usual information-gain regret behavior, whereas zero-jitter PI is more exploitative and admits a weaker worst-case guarantee.

Motivated by our theoretic findings, we propose a prequential procedure for selecting $\alpha$ online: it decreases $\alpha$ when realized prediction errors exceed model-implied uncertainty and returns $\alpha$ toward one as calibration improves. Empirical results demonstrate that tempering provides a practical yet theoretically grounded tool for stabilizing BO surrogates under localized sampling. 
\end{abstract}


\noindent%
{\it Keywords:} Model misspecification; Regret analysis; Online Calibration; generalized expected improvement; Nonparametric Regression.
\vfill

\newpage
\spacingset{1.2} 
\section{Introduction}

\subsection{Basic concepts and BO pipeline}\label{sec:intro:basics}

Bayesian optimization (BO) aims to identify a global maximizer of an unknown black-box function $f$ defined on a compact search domain $\mathcal X$ \citep{garnett_bayesoptbook_2023}. The algorithm proceeds by sequentially querying an input location $x_t\in\mathcal X$ and observing a noisy response $y_t=f(x_t)+\varepsilon_t$ from the truth $f(\cdot)\in\mathbb{R}$ with noise $\varepsilon_t\sim\mathcal N(0,\sigma^2)$. The datum collected up to time $t$ is denoted by $\mathcal D_t=\{(x_s,y_s)\}_{s=1}^t$ and the procedure stops until time $T$. Based on the datum, we may build a surrogate model summarizes $\mathcal D_t$ and delivers a predictive mean $\mu_t(\cdot;\theta)$ and a predictive standard deviation $\sigma_t(\cdot;\theta)$, where $\theta$ usually denotes model hyperparameters estimated from $\mathcal{D}_t$. 

In this work, the surrogate may be a linear model or a Gaussian process (GP) regression. When the algorithm proceeds, the next location $x_{t+1}$ is selected by maximizing an acquisition function $\mathcal{A}$ that balances exploration with exploitation and is measurable with respect to $\mathcal D_t$.  The choice of the next location comes from the rule
$x_{t+1}=\arg\max_{x\in\mathcal X}\,\mathcal A(x\mid\mathcal D_t)$,
where $\mathcal A$ is any acquisition computed from the current surrogate. 
Figure \ref{fig:schema} presents a schematic view of the pipeline where we allow the surrogate update step to employ a \emph{tempered posterior}, in which the likelihood is raised to a fractional power $\alpha\in(0,1]$ \citep{bhattacharya2019bayesian}. When $\alpha_t\equiv 1$, this procedure is identical to the usual BO pipeline \citep{garnett_bayesoptbook_2023}. 

\begin{figure}
    \centering
    \IfFileExists{figs/flow_chart.pdf}{%
      \includegraphics[width=0.995\linewidth]{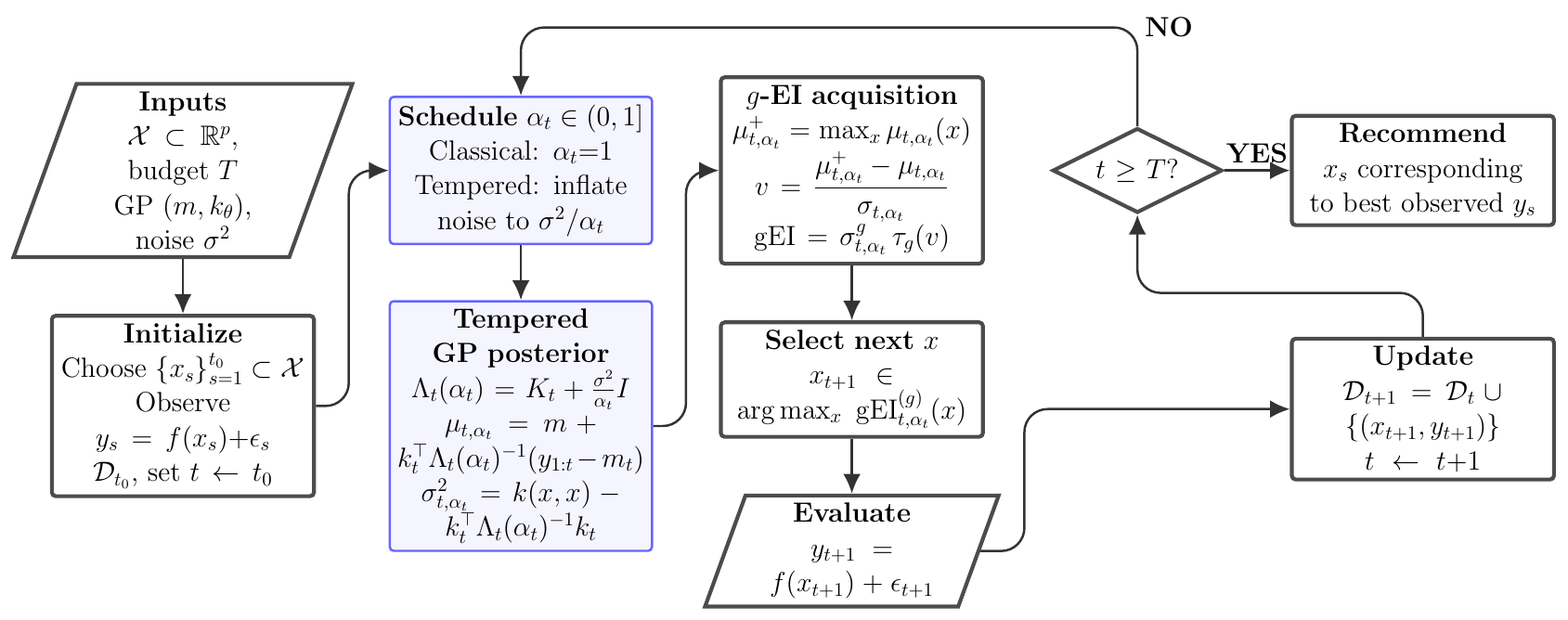}%
    }{%
      \fbox{\parbox[c][2in][c]{0.95\linewidth}{\centering
        Missing figure asset:\\
        \texttt{\detokenize{figs/flow_chart.pdf}}}}%
    }
    \caption{Schematic procedure of BO with possible tempered posterior. The tempered posterior part is indicated by blue boxes, setting $\alpha_t\equiv1$ in these blue boxes reduces to the regular BO.}
    \label{fig:schema}
\end{figure}

\subsection{Related literature}\label{sec:intro:literature}
The foundations of BO trace back to early work on expected improvement and related strategies that use GP surrogates to guide adaptive evaluation \citep{Mockus1975,gramacy2020surrogates}. Our work focuses on the generalized expected improvement ($g$-EI) acquisition functions that contain probability of improvement (PI) and classical expected improvement (EI) as special cases\citep{schonlau1998global}. These criteria provide a convenient unifying framework for studying how sensitivity to improvement magnitude influences the exploration-exploitation tradeoff. We note that subsequent development in BO also introduced alternative acquisition principles, including knowledge gradient policies and information based criteria \citep{frazier2009knowledge, srinivas2012information}. Parallel to this line of work, a substantial BO theory has established cumulative regret guarantees for various acquisition functions, with GP-UCB \citep{srinivas2012information} serving as a canonical example and motivating many subsequent refinements \citep{wang2014theoreticalanalysisbayesianoptimisation, jaiswal2023generalized}.

A central theme of this paper is the robustness of BO updating under model misspecification. One widely studied strategy replaces the usual Bayesian posterior with a likelihood tempered update,
\begin{align}
p_\alpha(\theta\mid \mathcal{D})\ \propto\ p(\theta) p(\mathcal{D}|\theta)^\alpha,\qquad \alpha\in(0,1]\label{eq:alpha_post},
\end{align}
which downweights the information contributed by each observation. The same construction appears under several names across communities. In the robust Bayesian asymptotics literature it is commonly referred to as a \emph{tempered} (or \emph{fractional}) posterior \citep{pitas2024fine,alquier2020concentration,yang2020alphavi,bhattacharya2019bayesian}, whereas in computational Bayes and marginal-likelihood estimation it is often called a \emph{power posterior} and
used for thermodynamic integration and related evidence estimators \citep{friel2008power,watanabe2013wbic}. Throughout, we use the term \emph{tempered posterior} to emphasize the role of $\alpha$ as a temperature/learning-rate parameter.

This perspective connects directly to generalized Bayes, which replaces the log-likelihood with a loss
and introduces a learning-rate parameter that controls the update magnitude \citep{bissiri2016general,lyddon2019general,syring2019calibrating,knoblauch2022optimization}. Growing Bayes-PAC theory clarifies why tempering helps under misspecification, from misspecified Bernstein-von Mises guarantees to concentration results for tempered posteriors and their variational approximations \citep{kleijn2012bvm,alquier2020concentration,yang2020alphavi}. Practical schemes for choosing or adapting the temperature include SafeBayes and related extensions \citep{grunwald2012safe,deHeide2020safe, grunwald2017inconsistency}, and following works explored exact conditioning with coarsened or divergence based updates, delivering robustness to small departures from the model \citep{miller2019coarsening,ghosh2016dpd}.

In terms of BO, to the best of our knowledge, we develop the first tempered-posterior regret analysis for improvement-based GP Bayesian optimization, together with a BO-specific prequential temperature schedule. A related line of work adapts exploration-exploitation behavior by modifying the acquisition rule itself \citep{golovin2017google,srinivas2009gaussian}.  Classic extensions include noisy and augmented variants of expected improvement that introduce exploration margins or corrections for observation noise \citep{Huang2006AEI,Letham2019ConstrainedBO,hu2025adjusted}. Other approaches adaptively combine multiple acquisitions 
\citep{Brochu2010GPHedge,Shahriari2014ESP}, develop GP updates for observation corruption \citep{ezzerg2026robustbayesianoptimisationunbounded}, or adjust internal weighting parameters within improvement-based rules to rebalance exploration over time \citep{Benjamins2023SAWEI,Qin2017TTEI}. Recent work revisits acquisition design from a variational-inference perspective, proposing hybrid acquisition rules \citep{cheng2025unified}. These methods regulate exploration through acquisition-level tuning or model combination, whereas our approach operates at the surrogate-update level via likelihood tempering, providing a complementary and broadly applicable mechanism for robustly alleviating posterior overconfidence.

\subsection{Contributions and Organization}\label{sec:intro:organization}

Despite these developments illustrated in Section \ref{sec:intro:literature}, existing regret analyses for BO typically assume the standard Bayesian update ($\alpha=1$), while the robust Bayes literature largely studies tempering outside sequential decision problems. The present paper addresses this gap by analyzing BO under tempered surrogates and deriving regret bounds that make the roles of tempering and the choice of acquisition functions within the generalized improvement family explicit.

Building on these strands, we study BO with tempered surrogate updates and characterize how tempering affects improvement-based acquisition functions. Our contributions are threefold. First, we embed likelihood tempering into Bayesian GP surrogates and demonstrate how it can stabilize predictive uncertainty and sequential optimization under localized model misspecification. Second, for the generalized improvement family indexed by $g$, we derive an analytic representation of the tempered acquisition and establish finite-time cumulative regret bounds for fixed-$\alpha$, confidence-rescaled $g$-EI. Our analysis accommodates predictable,
time-varying kernel hyperparameters through a uniform self-normalized confidence argument and an explicit kernel-domination condition. The resulting bounds distinguish the effects of tempering and the improvement order through the tempered information gain and acquisition-specific
constants. Third, we propose a tuning-light prequential schedule for choosing $\alpha_t$ and prove its calibration limit under explicit martingale and noise-estimation conditions. The adaptive schedule is an empirical extension; a finite-time regret theorem for adaptive $\alpha_t$ remains open.

In the rest of the paper, Section \ref{sec:motivations:background} motivates robustness through the tempered posterior and Section \ref{sec:pi:empirical} provides a concrete toy illustration. Section \ref{sec:GP-ei-alpha} develops GP surrogates, gives a closed form for the generalized expected improvement, and provides regret bounds as our main theoretic investigation of the tempered posterior in BO setting. Section \ref{sec:algorithmDesign} describes algorithmic design choices, including a tuning-light schedule for adaptively choosing $\alpha$. Section \ref{sec:exps} reports experimental results on benchmark functions, followed by a concluding discussion in Section \ref{Discussion}. Appendix \ref{sec:linear-ei-alpha} provides a complementary analysis for tempered Bayesian linear surrogates. Additional proofs, implementation details, and experimental results are collected in the remaining appendices.

\section{Motivations} 

\subsection{Misspecification and Regret Guarantees}\label{sec:motivations:background}

The BO pipeline above highlights a tension between statistical modeling and sequential decision-making. Since BO produces an adaptive design \citep{srinivas2009gaussian} and often concentrate evaluations near the current incumbent, the surrogate can become predictively overconfident in the sense that its model-implied predictive uncertainty understates realized errors precisely in the region that drives acquisition decisions. 

The miscalibration we target is primarily \emph{uncertainty misspecification} rather than global correctness of the surrogate class.  It can arise from kernel mismatch (e.g., incorrect smoothness/lengthscale or stationarity imposed on a nonstationary function), noise misspecification (e.g., underestimating $\sigma^2$, heteroskedasticity, or heavy-tailed errors), and repeated local refitting on a highly nonuniform design. A standard Bayesian device to alleviate misspecification is the tempered posterior \citep{bhattacharya2019bayesian}, which raises the likelihood to a power $\alpha\in(0,1]$ and thereby downweights the information contributed by the observations. For Gaussian noise GP regression, this is equivalent to inflating the effective noise variance from $\sigma^2$ to $\sigma^2/\alpha$, preventing the posterior from collapsing too quickly. Thus, tempering acts directly at the surrogate update step that produces the predictive mean and variance used by the acquisition.

Although tempering may mitigate predictive overconfidence, its implications for sequential optimization are not immediate. In particular, it is unclear how the additional regularization affects cumulative regret or how it interacts with the acquisition function. At the same time, the choice of acquisition function plays a central role in BO and can directly influence regret. For instance, PI is more exploitative and can be effective once a promising region has been identified, whereas EI offers a more balanced exploration-exploitation trade-off. To study this interaction between acquisition-level and surrogate-level adaptiveness systematically, we work with widely used generalized EI ($g$-EI) family \citep{schonlau1998global}:
\begin{equation} 
\alpha_{\theta,g,\alpha}^{EI(f)}(x\mid\mathcal D_t)
=
\mathbb E\Big[\big(\max\{0,\,f(x)-\mu^+_{t,\alpha}(\theta)\}\big)^g \,\Big|\,\mathcal D_t\Big],
\qquad g\in\mathbb R_{\ge 0}, \qquad \alpha \in (0,1],
\label{eq:g-ac-alpha-fixed}
\end{equation}
where $\mu_{t,\alpha}^{+}(\theta):=\max_{x\in\mathcal{X}}\mu_{t, \alpha}(x;\theta)$. For $g=0$ we use the indicator convention $(z_+)^0:=\mathbf 1\{z>0\}$, so the acquisition is a probability of improvement. Practically, we sometimes replace the threshold by $\mu_{t,\alpha}^{+}(\theta)+\xi$ with jitter $\xi>0$. 
The algorithm selects the next query point by computing $x_{t+1}=\arg\max_{x\in\mathcal{X}}\alpha_{\theta,g,\alpha}^{EI(f)}(x|D_{t})$.
In the standard posterior setting ($\alpha=1$), equation (\ref{eq:g-ac-alpha-fixed}) reduces to zero-jitter PI when $g=0$, and it recovers classical EI when $g=1$. Because the threshold is the posterior-mean maximum, zero-jitter PI is exactly posterior-mean greedy: every posterior-mean maximizer attains PI $1/2$, and no other point attains a larger value. Positive jitter removes this degeneracy. Larger $g$ values place greater weight on the magnitude of improvement. In this formulation, $\alpha$ and $g$ play distinct and complementary roles: $\alpha$ controls how aggressively the surrogate trusts the data, while $g$ controls how the acquisition maps posterior uncertainty into exploration incentives. This family of acquisition functions \eqref{eq:g-ac-alpha-fixed} is moment-based and is distinct from entropy-based acquisition functions studied in \citet{cheng2025unified}.

Existing BO regret analyses typically assume a standard Bayesian posterior update ($\alpha=1$) and do not directly apply to a tempered surrogate, since the predictive variances and the associated information-gain quantities are altered. Motivated by this theoretical gap, we analyze generalized improvement under a tempered GP predictive distribution and derive regret bounds that make the respective roles of $\alpha$ and $g$ explicit. The analysis clarifies when tempering reduces the noise-driven component of the guarantee, while also identifying the deterministic regularization trade-off that prevents arbitrarily aggressive tempering from being uniformly beneficial.  We further propose a simple prequential schedule for selecting $\alpha$ in a plug-in GP BO implementation as a novel methodology inspired by our theoretical result.

\subsection{A Toy Illustration}\label{sec:pi:empirical}

To provide an intuitive illustration of why likelihood tempering can matter in BO, we consider a simple one-dimensional example designed to highlight a common failure mode of PI: when the surrogate becomes overconfident early, PI can concentrate on a single basin and under-explore competing regions.

Figure \ref{fig:one-d-pi-tempered} visualizes the following oscillatory black-box function, whose symmetric global maximizers are approximately $x=0.4032$ and $x=0.5968$, with maximum value approximately $0.9294$:
\begin{align}
f(x)\;=\;-2\,\frac{\cos\!\left(8\,|4x-2|\right)}{|4x-2|^{2}+2},
\quad x\in[0,1],\label{eq:blackbox_fun}
\end{align}
observed with Gaussian noise of standard deviation $0.05$. We fit a GP surrogate with a Mat\'ern covariance (smoothness fixed to $\nu=2$) and estimate the remaining
hyperparameters, including a white-noise term by marginal likelihood (with multiple restarts). We run PI with jitter $\xi=0.01$ under three tempering levels $\alpha\in\{0.1,0.5,1.0\}$, where $\alpha=1$ is the usual posterior. Each choice of $\alpha$ starts from the same five initial points and then repeatedly selects
$x_{t+1}$ by maximizing PI under the corresponding tempered predictive distribution.

The key mechanism is that tempering reduces the effective influence of each observation in the surrogate update. Since PI depends on standardized improvement, an underestimated predictive variance can make PI sharply peaked around the current incumbent and discourage exploration elsewhere; tempering counteracts this by maintaining non-negligible uncertainty away from the incumbent and keeping alternative regions competitive.

This behavior is visible in Figure \ref{fig:one-d-pi-tempered}. By iteration $10$, the tempered run with $\alpha=0.1$ has already sampled the region around the global peak and attains a best observed value near $0.968$, while the $\alpha\in\{0.5,1.0\}$ trajectories remain around $0.819$ and $0.817$. The corresponding PI curves show that smaller $\alpha$ produces broader, less brittle acquisition profiles, whereas larger $\alpha$ yields more concentrated peaks that can reinforce repeated sampling in the same basin. This example illustrates that tempering can restore exploration to PI when the surrogate becomes locally overconfident under adaptive sampling. Additional toy variants are reported in Appendix \ref{sec:additional_toy_result}.

The example uses positive-jitter PI and is therefore illustrative rather than directly covered by the zero-jitter PI theorem in the next section. Our theory in Section \ref{sec:GP-ei-alpha} instead quantifies the fixed-$\alpha$ trade-off for confidence-rescaled improvement acquisitions. The example also motivates the practical question of choosing $\alpha$ in a sequential BO loop, which we address in Section \ref{sec:algorithmDesign} with a prequential information-matching schedule.

\begin{figure}[t]
\centering
\IfFileExists{figs/1d_tempered_example_big.png}{%
  \includegraphics[width=0.98\linewidth]{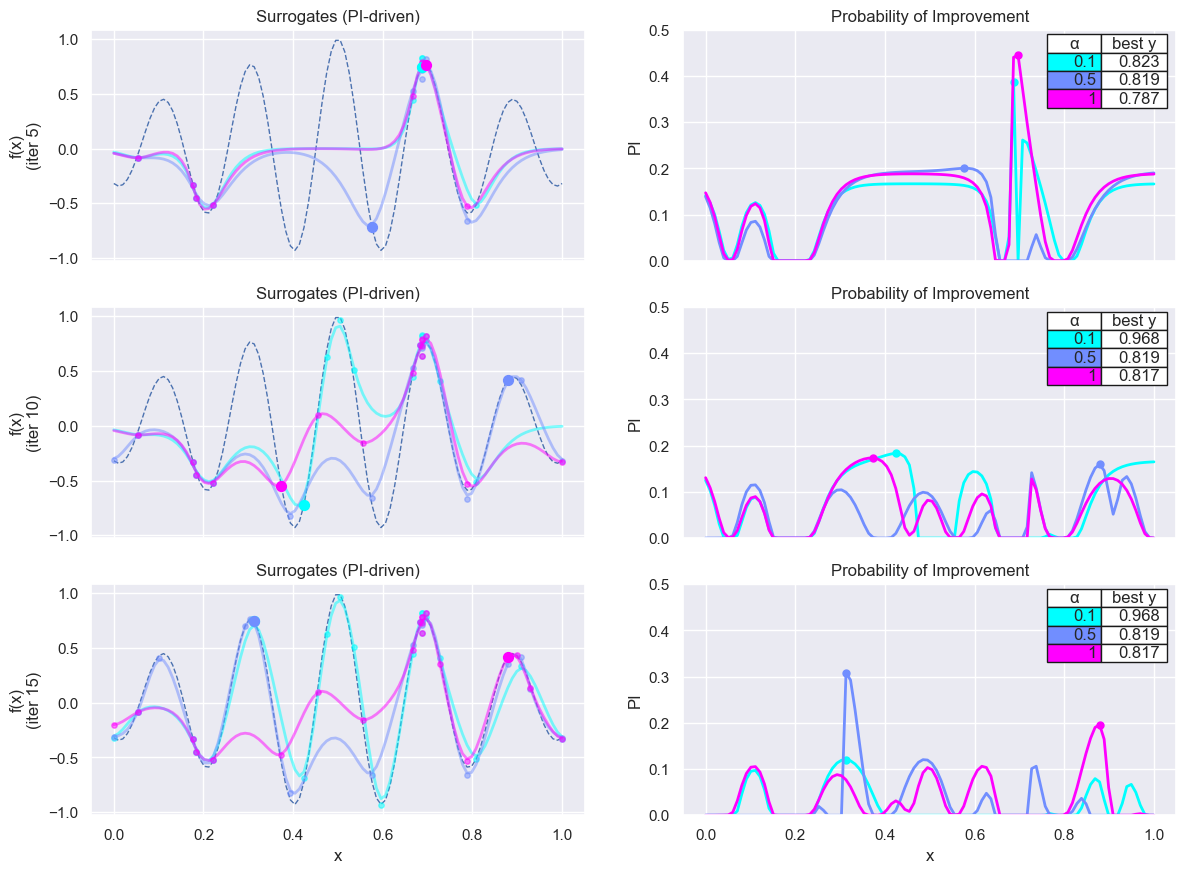}%
}{%
  \fbox{\parbox[c][2.5in][c]{0.94\linewidth}{\centering
    Missing figure asset:\\
    \texttt{\detokenize{figs/1d_tempered_example_big.png}}}}%
}
\caption{Tempered probability of improvement in one dimension. Each row shows iteration $t\in\{5,10,15\}$. The left panels show the true function in \eqref{eq:blackbox_fun} as a dashed curve and the posterior mean for each $\alpha\in\{0.1,0.5,1.0\}$ as a solid curve.  Bigger dots are the current candidates, smaller dots are the selected historical candidates. 
The table reports the best observed value of $y$ so far for each $\alpha$. Right panels plot the probability of improvement curves with their maximizers marked. The surrogate uses a Mat\`ern kernel with $\nu=2$ and a white noise term, observations have standard deviation $0.05$, 
five shared random initial points are used at iteration 0.}
\label{fig:one-d-pi-tempered}
\end{figure}

\section{Bayesian GP Surrogates}\label{sec:GP-ei-alpha}

In practice, the predominant approach in BO is to start by fitting a standard GP regression on the current data collection $x_{1:t}$ and use it as a surrogate for the unknown function $f$ \citep{gramacy2020surrogates}:
\[
f(x_{1:t})|\theta\sim \mathcal{N}(m(x_{1:t}),K^{\theta}(x_{1:t},x_{1:t})),
\]
where $m(.)$ is the mean function, and $\theta$ parameterizes the covariance kernel, with
$\bigl[K^{\theta}(x_{1:t},x_{1:t})\bigr]_{ij}= k^{\theta}(x_i,x_j)$. Without loss of generality, we consider a zero-mean function prior, and the covariance kernel is parametrized by a $p$-dimensional length-scale vector $\theta$ estimated by maximum likelihood estimates (MLE). We focus on covariance kernels from the Mat\'ern family, with smoothness parameter $\kappa>0$. Let
$r^2 = (\mathbf{x}-\mathbf{x}')^\top \operatorname{diag}(\theta^2)^{-1}(\mathbf{x}-\mathbf{x}').$
The Mat\'ern--$5/2$ kernel and the squared-exponential kernel, which arises as a limiting case of the Mat\'ern family under the conventional scaling,
are given by
$$k^{\theta}_{\text{SE}}(\mathbf{x}, \mathbf{x}')  = \exp\!\left(-\tfrac{1}{2} r^2\right), \quad  k_{5/2}^\theta(\mathbf{x}, \mathbf{x}') 
= \exp(-\sqrt{5}\, r)
\left( 1 + \sqrt{5}\, r + \tfrac{5}{3} r^2 \right).$$

\subsection{Tempered Posteriors}
Consider the case where we update the GP surrogate model using the tempered posterior with $\alpha>0$. After $t$ steps with dataset $\mathcal{D}_{t}=\{(x_{s},y_{s})\}_{s=1}^{t}$,
we define $x_{1:t}:=[x_{1},\cdots,x_{t}]$, $y_{1:t}:=[y_{1},\cdots,y_{t}]$,
$K_{t}^{\theta}:=[k^{\theta}(x,x')]_{x,x'\in x_{1:t}}$, and $f:=[f(x_1),\cdots,f(x_t)]$. For a nonsingular $K_t^{\theta}$ (with the usual Gaussian-conditioning extension when $K_t^{\theta}$ is singular), the tempered posterior can be derived as
\begin{align*}
p_{\alpha}(f|\mathcal{D}_{t}) & \propto p(f)(p(y_{1:t}|f))^{\alpha}\\
 & \propto\exp\left \{-\frac{1}{2}f^{T}(K_{t}^{\theta})^{-1}f \right \}\\
 & \qquad{}\times\exp\left \{-\frac{1}{2}(y_{1:t}-f)^{T} \left(\frac{\sigma^{2}}{\alpha}\mathbb{I}\right)^{-1}(y_{1:t}-f)\right \}\\
 & \sim N(\mu_{t,\alpha},\Sigma_{t,\alpha}),
\end{align*}
where 
 $\mu_{t,\alpha}=K_{t}^{\theta} \left (K_{t}^{\theta}+\frac{\sigma^{2}}{\alpha}\mathbb{I} \right)^{-1}y_{1:t}$, 
 $\Sigma_{t,\alpha}=K_{t}^{\theta}-K_{t}^{\theta} \left (K_{t}^{\theta}+\frac{\sigma^{2}}{\alpha}\mathbb{I}_t \right)^{-1}K_{t}^{\theta}$.
To evaluate an arbitrary point $x\in \mathcal{X}$ using the updated surrogate
model, we define $k_{t}^{\theta}(x)=[k^{\theta}(x_{1},x),\cdots,k^{\theta}(x_{t},x)]^{T}$, and
obtain 
\begin{align*}
 & \mu_{t,\alpha}(x; \theta)=k_{t}^{\theta}(x)^{T} \left (K_{t}^{\theta}+\frac{\sigma^{2}}{\alpha}\mathbb{I}_t \right)^{-1}y_{1:t},\\
 & k_{t,\alpha}^{\theta}(x,x')=k^{\theta}(x,x')-k_{t}^{\theta}(x)^{T} \left(K_{t}^{\theta}+\frac{\sigma^{2}}{\alpha}\mathbb{I}_t \right)^{-1}k_{t}^{\theta}(x'),\\
 & \sigma^2_{t,\alpha}(x; \theta)=k_{t,\alpha}^{\theta}(x,x).
\end{align*}

\subsection{The \texorpdfstring{$g$}{g}-EI Acquisition Function}

Our regret analysis covers the entire $g$-EI family for $g \geq 0$, which includes both probability of improvement $(g=0)$ and classical expected improvement $(g=1)$. Under the $\alpha$-tempered GP posterior, the latent objective at any fixed $x$ has the predictive distribution
$f(x)\mid \mathcal D_t,\alpha,\theta \sim \mathcal N\!\big(\mu_{t,\alpha}(x;\theta),\,\sigma^2_{t,\alpha}(x;\theta)\big),$
and we measure improvement relative to the current posterior-mean maximizer $\mu^{+}_{t,\alpha}(\theta)\;=\;\max_{x\in\mathcal X}\mu_{t,\alpha}(x;\theta)$. The following proposition gives a convenient analytic representation of the $g$-EI family defined in \eqref{eq:g-ac-alpha-fixed}.:

\begin{proposition} \label{prop:tau-g}
For $\alpha\in(0,1]$ and $g\geq 0$, Let $v_{\alpha}=\frac{\mu_{t,\alpha}^{+}(\theta)-\mu_{t, \alpha}(x;\theta)}{\sigma_{t, \alpha}(x;\theta)}$. The $g$-EI family defined in (\ref{eq:g-ac-alpha-fixed}) can be rewritten as:
\begin{equation}
\alpha^{\mathrm{EI}(f)}_{\theta,g,\alpha}(x\mid\mathcal D_t)
=\sigma_{t,\alpha}^g(x;\theta)\,\tau_g\!\big(v_{\alpha}\big),
\qquad
\tau_g(v):=\int_v^\infty (u-v)^g\,\phi(u)\,du.
\label{eq:tau-integral}
\end{equation}
If $g\in\mathbb N_{\geq0}$, this representation reduces to the finite sum:
\begin{equation*}
\alpha^{\mathrm{EI}(f)}_{\theta,g,\alpha}(x\mid\mathcal D_t)
=\sigma_{t,\alpha}^g(x;\theta)\sum_{k=0}^{g}(-1)^{k}\binom{g}{k}\,
v_{\alpha}^{k}\,T_{g-k}\!\big(v_{\alpha}\big),
\label{eq:alpha-EI-g}
\end{equation*}
where $T_0(v)=\Phi(-v)$, $T_1(v)=\phi(v)$, and for $m>1$, $T_m(v)=v^{m-1}\phi(v)+(m-1)T_{m-2}(v)$.
\end{proposition}

Following \citet{wang2014theoreticalanalysisbayesianoptimisation}, we introduce a positive, round-dependent parameter $\nu_t$ that controls exploration by rescaling the posterior standard deviation within the acquisition function. Equivalently, the acquisition is evaluated using the same posterior mean and the rescaled standard deviation $\nu_t\sigma_{t,\alpha}(x;\boldsymbol{\theta})$. Combining this rescaling with Proposition~\ref{prop:tau-g}, we obtain
\begin{equation}
\alpha_{\boldsymbol{\theta},g,\alpha}^{EI(f)}(x\mid\mathcal D_t)=\nu_t^g
\sigma_{t,\alpha}^g(x;\boldsymbol{\theta})\tau_g\!\left(\frac{v_\alpha}{\nu_t}\right).
\label{eq:scaled-g-ac-alpha}
\end{equation}
Here, $v_\alpha$ is defined in Proposition~\ref{prop:tau-g}, and $\tau_g$ is defined in \eqref{eq:tau-integral}. For integer $g$, the finite-sum representation in Proposition~\ref{prop:tau-g} may be used.

\subsection{Regret Analysis of Tempered GP Surrogate}\label{subsec:GP_fractional}

Our regret analysis uses the tempered analogue of the \emph{maximum information gain}. For $t\geq1$, define
\begin{equation*}
\gamma_{t,\alpha}^{\boldsymbol{\theta}} := \sup_{(x_1,\ldots,x_t)\in\mathcal X^t}
\frac{1}{2}
\log \left| \mathbb{I}_t+\alpha\sigma^{-2} K_t^{\boldsymbol{\theta}}(x_{1:t})\right|.
\label{eq:mig-alpha}
\end{equation*}
When $\alpha=1$, this reduces to the standard maximum information gain that appears throughout the BO and GP-bandit regret literature \citep{garnett_bayesoptbook_2023,srinivas2012information}. Our analysis
uses a self-normalized concentration argument for predictable sequential designs. Under Assumption \ref{ass:uniform-theta}, Proposition \ref{prop:uniform-alpha-confidence} establishes a high-probability confidence event for the tempered GP surrogate that holds uniformly over time, design points, and admissible kernel hyperparameters. Proofs of the results in this subsection are deferred to
Appendix~\ref{sec:GP=proof}.

\begin{assumption}[Uniform hyperparameter regularity]
\label{ass:uniform-theta}
The design is predictable: $x_t$ is $\mathcal F_{t-1}$-measurable, and $y_t=f(x_t)+\varepsilon_t$, where $\varepsilon_t$ is conditionally
$\sigma$-subGaussian given $\mathcal F_{t-1}$. The learned hyperparameters
satisfy $\theta_t\in\Theta$ and are $\mathcal F_{t-1}$-measurable, where
$\Theta \subset\mathbb R^p$ is compact. The kernels are normalized so that
$k^\theta(x,x)=1$ for all $x\in\mathcal X$ and $\theta\in\Theta$. Define
$B_\Theta:=\sup_{\theta\in\Theta}\|f\|_{\mathcal H_\theta(\mathcal X)}<\infty$, and $\overline\gamma_{t,\alpha}^{\Theta}:= \sup_{\theta\in\Theta}\gamma_{t,\alpha}^{\theta}$.

Assume also a uniform posterior-variance domination condition: there exist a reference hyperparameter $\theta^0\in\Theta$ and a constant $C_{\mathrm{var}}<\infty$ such that, for every $n\geq0$, every predictable design $x_{1:n}$, every $x\in\mathcal X$, every $\theta\in\Theta$, and every $\alpha\in(0,1]$,
\[
\sigma_{n,\alpha}^2(x;\theta)\leq C_{\mathrm{var}}\,\sigma_{n,\alpha}^2(x;\theta^0).
\]
For ordered length-scale  Mat\'ern families, this condition can be verified by a kernel-specific comparison argument. 

For every finite horizon $T$ and every $\rho\in(0,1)$, there exist deterministic envelopes $L_{t-1,\alpha}^{\mu}(\rho)$ and $L_{t-1,\alpha}^{\sigma}(\rho)$ and an event
$\mathcal E_{T,\rho}^{\mathrm{Lip}}$ with probability at least $1-\rho$ on
which $\mu_{t-1,\alpha}(\cdot;\theta)$ and
$\sigma_{t-1,\alpha}(\cdot;\theta)$ are uniformly Lipschitz in $\theta$:
for all $1\le t\le T$, $x\in\mathcal X$, and
$\theta,\theta'\in\Theta$,
\begin{align*}
|\mu_{t-1,\alpha}(x;\theta)-\mu_{t-1,\alpha}(x;\theta')|
&\le L_{t-1,\alpha}^{\mu}(\rho)\|\theta-\theta'\|,\\
|\sigma_{t-1,\alpha}(x;\theta)-\sigma_{t-1,\alpha}(x;\theta')|
&\le L_{t-1,\alpha}^{\sigma}(\rho)\|\theta-\theta'\|.
\end{align*}
\end{assumption}

\begin{proposition}[Uniform confidence for tempered GP surrogates]
\label{prop:uniform-alpha-confidence}
Suppose Assumption \ref{ass:uniform-theta} holds. Fix $T\ge1$, $\delta\in(0,1)$, and $\alpha\in(0,1]$. Then there exists a deterministic sequence of confidence radii $\{\Phi_{t,\alpha}^{\Theta}\}_{t=1}^T$, chosen before observing the data, such that, with probability at least $1-\delta$, simultaneously for all $1\le t\le T$, all $x\in\mathcal X$, and all $\theta\in\Theta$,
$$
|\mu_{t-1,\alpha}(x;\theta)-f(x)|
\le
\Phi_{t,\alpha}^{\Theta}
\sigma_{t-1,\alpha}(x;\theta).
$$
\end{proposition}

Our regret bounds require the acquisition rescaling parameter $\nu_t$ to be of the same order as the confidence radius, as formalized in Assumption~\ref{ass:nu-scaling}. Similar confidence-scale rescalings have appeared in earlier EI regret analyses \citet{wang2014theoreticalanalysisbayesianoptimisation}, but our proof below uses the self-normalized confidence event in Proposition \ref{prop:uniform-alpha-confidence}.

\begin{assumption}[Acquisition rescaling]
\label{ass:nu-scaling}
Let $\Phi_{t,\alpha}^{\Theta}$ be the confidence radius from Proposition~\ref{prop:uniform-alpha-confidence}. There exist constants $0<c_\nu\le C_\nu<\infty$, such that $c_\nu\Phi_{t,\alpha}^{\Theta} \le \nu_t \le C_\nu\Phi_{t,\alpha}^{\Theta}$ uniformly over $1\le t\le T$.
\end{assumption}

We now state the main cumulative regret bound for tempered  positive-order $g$-EI.

\begin{theorem}[Regret bound for tempered $g$-EI] \label{thm:general-g-alpha}
Fix $g > 0$, $\alpha\in(0,1]$, $T\ge1$,  and $\delta\in(0,1)$. Suppose Assumptions~\ref{ass:uniform-theta} and \ref{ass:nu-scaling} hold. Let $\Phi_{T,\alpha}^{\Theta,\max}:=\max_{1\le t\le T}\Phi_{t,\alpha}^{\Theta}$ and $D_{g,\nu}
:=\sup_{0\le z\le 1/c_\nu} \frac{1/c_\nu-z}{\tau_g(z)^{1/g}}$.
Then, with probability at least $1-\delta$,
$$
R_T = \mathcal O\!\left( \Phi_{T,\alpha}^{\Theta,\max} \kappa_{T,\alpha,g} \sqrt{\frac{\overline\gamma_{T,\alpha}^{\Theta}T}{\log(1+\alpha \sigma^{-2})}} \right),
$$
where the hidden constant may depend on $C_{\mathrm{var}}$ but not on $T$, and
$$\kappa_{T,\alpha,g}:= 1+D_{g,\nu}C_{\nu}C_{(g)}^{1/g}+C_{\nu}\tau_g^{-1}\!\left[C_{(g)}
\left(\frac{\sigma^2}{\alpha(T-1)+\sigma^2}\right)^{g/2}\right],
\qquad C_{(g)}:=\frac{2^{g/2}\Gamma((g+1)/2)}{2\sqrt{\pi}}.$$

\end{theorem}

\begin{corollary}[Order of the regret bound] \label{cor:general-g-alpha-order}
Under the conditions of Theorem~\ref{thm:general-g-alpha}, assume in addition that, for the fixed $\alpha$ and $\delta$, there exist constants $C_L,q_L<\infty$ such that for every $T\geq1$,
\[
\max_{1\leq t\leq T}\left\{L_{t-1,\alpha}^{\mu}(\delta/2)+L_{t-1,\alpha}^{\sigma}(\delta/2)\right\}\leq C_L T^{q_L}.
\]
Then, with probability at least $1-\delta$,
$$
R_T = \mathcal O\!\left( \left[ B_\Theta+ \sqrt{ \alpha\left( \overline\gamma_{T,\alpha}^{\Theta} +p\log T+\log(1/\delta) \right)}\right]
\kappa_{T,\alpha,g} \sqrt{\frac{\overline\gamma_{T,\alpha}^{\Theta}T}{\log(1+\alpha \sigma^{-2})}}  \right).
$$
\end{corollary}

Corollary \ref{cor:general-g-alpha-order} clarifies that tempering improves the stochastic and information-gain components of the bound, but excessive tempering can be limited by the deterministic RKHS term. To see this, let $A_{T,\alpha}:=\overline\gamma_{T,\alpha}^{\Theta}+p\log T+\log(1/\delta)$ and $\ell_\alpha:=\log(1+\alpha\sigma^{-2})$, so that the bound can be rewritten as $\kappa_{T,\alpha,g} \sqrt{T} \left[ \sqrt{A_{T,\alpha} \bar{\gamma}_{T,\alpha}^{\Theta} \frac{\alpha}{\ell_{\alpha}}}+B_\Theta\sqrt{\bar{\gamma}_{T,\alpha}^{\Theta}/\ell_\alpha}\right].$ The first term is the stochastic component. Since the terms $\kappa_{T,\alpha,g}$, $A_{T,\alpha}$, $\bar{\gamma}_{T,\alpha}^{\Theta}$, and $\alpha/\ell_\alpha$ are all nondecreasing in $\alpha$, this component improves with tempering (smaller $\alpha$). This supports the use of tempering when posterior overconfidence and stochastic prediction error dominate the regret. However, our result does not imply that excessive tempering ($\alpha \downarrow 0$) is uniformly better, due to the existence of the second deterministic RKHS term $B_\Theta \sqrt{\bar{\gamma}_{T,\alpha}^{\Theta}/\ell_\alpha}$. It follows that tempering is beneficial when stochastic uncertainty is the dominant part of the regret.

The dependence on $g$ is governed by two opposing terms in $\kappa_{T,\alpha,g}$. In Appendix \ref{subsec:g-monotonicity}, we show that $D_{g,\nu}C_{(g)}^{1/g}$ is nonincreasing in $g$ and diverges as $g\downarrow0$. In contrast, the $\tau_g^{-1}(.)$ term is nondecreasing in $g$ and diverges as $g \rightarrow \infty$. These opposing limits imply that the acquisition-specific factor is minimized in an interior positive-$g$ regime, although the minimizing value depends on the problem parameters. The classical EI choice $g=1$ is therefore a natural interior benchmark, but the theorem does not identify it as universally optimal. We next analyze the $g=0$ case separately.

\begin{theorem}[Regret bound for probability of improvement] \label{thm:pi-alpha}
Fix $g=0$, $\alpha\in(0,1]$, $T\ge1$, and $\delta\in(0,1)$. Suppose Assumption \ref{ass:uniform-theta} holds. Then, with probability at least $1-\delta$,
$$R_T= \mathcal O\!\left( \Phi_{T,\alpha}^{\Theta,\max}\kappa_{T,\alpha,0} \sqrt{\frac{\overline\gamma_{T,\alpha}^{\Theta}T}{\log(1+\alpha \sigma^{-2})}}  \right),$$
where
$$\kappa_{T,\alpha,0}=1+\sqrt{\frac{\alpha(T-1)+\sigma^2}{\sigma^2}}.$$
\end{theorem}

Theorem \ref{thm:pi-alpha} concerns zero-jitter PI, which is posterior-mean greedy and therefore contains no direct posterior-variance exploration incentive. The result carries an additional order-$\sqrt{T}$ factor in $\kappa_{T,\alpha,0}$, which gives a weaker worst-case guarantee than the positive-order $g$-EI bound. This supports the standard caution that PI can be overly exploitative, but only as an upper-bound separation; it does not prove that PI is uniformly poor in practical finite-sample settings. However, the result does not analyze the positive-jitter PI used in Figure~\ref{fig:one-d-pi-tempered}, nor does it by itself imply monotone improvement as $\alpha$ decreases.

When $\alpha=1$ and $g=1$, Theorem~\ref{thm:general-g-alpha} has the standard information-gain structure familiar from GP bandit analysis: a
confidence width multiplied by $\sqrt{T\gamma_T}$, together with an
EI-specific inverse-$\tau_1$ factor arising from the improvement acquisition.
This comparison is most natural with the information-gain regret analyses of
GP-UCB and kernelized bandit algorithms \citep{srinivas2012information,Chowdhury17,vakili2021information}, rather
than with classical noiseless EI convergence theory. The latter typically
studies fixed-prior, noiseless expected improvement and simple-regret-type
convergence behavior \citep{bull2011convergence}. Recent noisy EI analyses emphasize that cumulative regret depends delicately on the choice of incumbent \citep{wang2025bayesian}. Our result is complementary: it provides a self-normalized, frequentist RKHS guarantee for tempered, confidence-rescaled $g$-EI with adaptively learned kernel hyperparameters. Algorithmically, the special case $\alpha=1$, $g=1$, and $\nu_t=1$ reduces to unscaled EI. The regret theorem covers this unscaled rule
when the choice $\nu_t=1$ also satisfies Assumption \ref{ass:nu-scaling}.

Our theoretical results give tempering a finite-sample calibration interpretation in adaptive surrogate optimization. Reducing $\alpha$ does more than inflate the posterior variance: it strengthens kernel-ridge regularization, and systematically reduces the noise-driven confidence component. However,  aggressive tempering can leave the deterministic RKHS regularization component dominant. Our results also connect with generalized-Bayes motivations for using a learning rate below one under misspecification, with frequentist-style regret guarantees for adaptive experimental design. The fixed-$\alpha$, confidence-rescaled analysis intentionally isolates this mechanism; extending the finite-time regret
theory to the data-dependent tempering schedule developed in Section~\ref{sec:algorithmDesign} remains open.

\section{Tempering Schedule Design}
\label{sec:algorithmDesign}

Our regret bounds in Section~\ref{sec:GP-ei-alpha} show that, for a fixed acquisition parameter $g$ and fixed confidence-rescaling constants, reducing $\alpha$ weakly tightens the noise-driven confidence component of the bound. The deterministic RKHS component, however, prevents a uniform ordering of
the full regret bound over $\alpha$. This trade-off motivates calibrating the tempering level rather than choosing $\alpha$ as small as possible. In practice, when researchers have access to a large external dataset $\mathcal{D}_n$, existing methods may be used to estimate an appropriate tempering level $\hat{\alpha}$ (e.g., \cite{holmes2017assigning,miller2019coarsening}), see \citet{wu2023comparison} for a comprehensive review of how to choose $\alpha$. However, in many Bayesian optimization settings, such calibration dataset is unavailable, and picking the ``optimal'' $\alpha$ remains challenging. We hence develop an adaptive strategy for choosing $\alpha$ based on information-matching principle: the tempering level is chosen so that the prior expected information gain from a single observation under a possibly misspecified model matches the gain one would obtain if the model were correct. Our idea of designing a tempering schedule follows \citet{holmes2017assigning}, who propose to choose $\alpha$ such that 
\begin{equation} \label{eq:hw-alpha}
    \alpha = \left \{\frac{\int f(x; \theta_0) \Delta(x) dx}{\int f_0(x) \Delta(x) dx } \right \}^{\frac{1}{2}},
\end{equation}
    where $f_0(x)$ is the true data density, $f(x; \theta_0)$ is the pseudo true model that minimizes the KL divergence to $f_0(x)$, and $\Delta(x)$ is the Fisher-divergence information between a posterior update from its prior, with likelihood $f(x;\theta)$. Our adaptation of this principle leads to a simple, tuning-free, and computationally efficient procedure for choosing $\alpha$ online.
    
 Fix a design point $x \in \mathcal{X}$ and write the local parameter of interest as $\theta = f(x)$, under the working observation model, we have $y|\theta \sim \mathcal{N}(\theta, \sigma^2)$. Following \citet{holmes2017assigning}, we quantify the information of a single update at $x$ as 
\begin{align*}
    \Delta_{t}(x;y) &:= \mathbb{E}_{\theta \sim N(\mu_{t-1,1}(x), \sigma_{t-1,1}^2(x)) } \left [\left(\frac{\partial }{\partial \theta} \log p(y|\theta)\right)^2 \right] \\
    & = \mathbb{E} \left [\left (\frac{(y-\theta)}{\sigma^2} \right)^2 \right] = \frac{(y-\mu_{t-1,1}(x))^2+ \sigma_{t-1,1}^2(x)}{\sigma^4}.
\end{align*}

Note that we have used $(\mu_{t-1,1}(x),\sigma_{t-1,1}^2(x))$ to denote the regular untempered GP predictive mean and variance given $\mathcal{D}_{t-1}$. Under a likelihood power $\alpha$, the score is multiplied by $\alpha$ and the squared score by $\alpha^2$. Consequently, the one-step information measure under tempering is $\Delta_{t,\alpha}(x;y)=\alpha^2\Delta_t(x;y)$.

We now take expectations of $\Delta_t(x;y)$ with respect to two different distributions for $y | x$. When our GP model is correctly specified, we have $y|x \sim N(\theta_0(x), \sigma^2)$, where $\theta_0(x)$ represents the true signal value at $x$. It follows $$\mathbb E_{\mathrm{model}}\!\left[\Delta_t(x;Y)\mid x\right]=\frac{\sigma_{t-1,1}^2(x)+\sigma^2+(\theta_0(x)-\mu_{t-1,1}(x))^2}{\sigma^4}.$$
When our model is misspecified, we let $y|x \sim f_0(.|x)$ with mean $m_0(x)$, variance $v_0(x)$. Then we have
$$\mathbb E_{0}\!\left[\Delta_t(x;Y)\mid x\right]
=
\frac{\sigma_{t-1,1}^2(x)+v_0(x)+(m_0(x)-\mu_{t-1,1}(x))^2}{\sigma^4}
=
\frac{\sigma_{t-1,1}^2(x)+\mathrm{MSE}_t(x)}{\sigma^4},$$
where $\mathrm{MSE}_t(x):=\mathbb E_0[(Y-\mu_{t-1,1}(x))^2\mid x]$. 

Following \citet{holmes2017assigning}, we choose a local tempering level $\alpha_t(x)$ so that the tempered information under the true world matches the information under the as-if-correct model world so that $\alpha_t(x)^2\;\mathbb E_{0}\!\left[\Delta_t(x;Y)\mid x\right]= \mathbb E_{\mathrm{model}}\!\left[\Delta_t(x;Y)\mid x\right].$ Therefore,

\begin{equation} \label{eq:alpha-local}
    \alpha_{t}^2(x) = \frac{\mathbb E_{\mathrm{model}}\!\left[\Delta_t(x;Y)\mid x\right]}{\mathbb E_{0}\!\left[\Delta_t(x;Y)\mid x\right]} = \frac{\sigma_{t-1,1}^2(x)+\sigma^2+(\theta_0(x)-\mu_{t-1,1}(x))^2}{\sigma_{t-1,1}^2(x)+v_0(x)+(m_0(x)-\mu_{t-1,1}(x))^2}.
\end{equation}
Under correct specification, $(m_0(x),v_0(x))=(\theta_0(x),\sigma^2)$, and
(\ref{eq:alpha-local}) gives $\alpha_t(x)=1$.

For implementation, it is impractical to specify different $\alpha$ values at each $x$ and we need to aggregate across visited points $x_t$ and use prequential quantities in the schedule.  Although $(m_0(x), v_0(x))$ is unknown, observe that by the standard bias-variance trade-off formula, $\text{MSE}(x):= \mathbb{E}[(y-\mu_{t-1,1}(x))^2 | x] =  v_0(x) + (m_0(x)-\mu_{t-1,1}(x))^2$. Based on $\mathcal{D}_t$, we compute the following prequential estimator of the global tempering weight at time $t$: 
\begin{align}\label{eq:hw-alpha-es}
\widehat{\alpha}_t = \min \left\{\sqrt{\frac{\frac{1}{t}\sum_{s=1}^t\!\big\{\sigma_{s-1,1}^2(x_s)+\widehat\sigma_t^2\big\}}
     {\frac{1}{t}\sum_{s=1}^t\!\big\{\sigma_{s-1,1}^2(x_s)+\widehat{\text{MSE}}_s\big\}}}, \quad  1\right \},
\qquad
\widehat{\text{MSE}}_s := \big(y_s-\mu_{s-1,1}(x_s)\big)^2,
\end{align}
where $\widehat\sigma_t^2$ denotes a global estimate of the observation-noise variance available at time $t$.
This estimator depends only on prequential quantities $(\mu_{s-1,1}(x_s),\sigma_{s-1,1}^2(x_s))$ and adopts a conservative information-matching approximation by setting the model-side bias term
$\big(\theta_0(x_s)-\mu_{t-1,1}(x_s)\big)^2$ to zero in the as-if-correct world.
In practice, \eqref{eq:hw-alpha-es} is easy to compute: after observing $y_t$, we compute $\widehat\alpha_t$ from $\mathcal D_t$ and use it to construct the tempered surrogate and select
$x_{t+1}$. This typically yields $\widehat{\alpha}_T\le 1$ early in BO, since $\widehat{\mathrm{MSE}}_t$ reflects both noise and misspecification; as posterior uncertainty contracts under correct specification and with a consistent $\widehat\sigma_t^2$, the ratio drifts toward $1$. We present our entire BO framework with scheduled tempered posterior in Algorithm \ref{alg:bo-alpha-gp-gei} of Appendix \ref{sec:algo-pseudo}. 

Our proposed tempering schedule in (\ref{eq:hw-alpha-es}) is highly general and remains coherent even when applied to acquisition functions beyond the $g$-EI family. We conclude this section with Proposition \ref{prop:calibration2}, which establishes that the proposed tempering schedule is well-behaved asymptotically. 
In the well-specified regime, the schedule satisfies $\widehat{\alpha}_t\xrightarrow{p}1$ under the stated conditions, indicating no tempering is needed. In contrast, when the model is misspecified, we have $\widehat{\alpha}_t$ converging to a limit smaller than $1$, reflecting the reduced information that each observation contributes relative 
to the misspecified surrogate.

\begin{proposition}[Limits of the Adaptive Schedule]\label{prop:calibration2}
Let $\{\mathcal{F}_s\}_{s\geq0}$ be the filtration generated by the BO history, and assume $y_s=f(x_s)+\varepsilon_s$, where
\[
\mathbb{E}[\varepsilon_s\mid\mathcal{F}_{s-1}]=0,\qquad
\mathbb{E}[\varepsilon_s^2\mid\mathcal{F}_{s-1}]=\sigma^2>0,\qquad
\mathbb{E}[\varepsilon_s^4\mid\mathcal{F}_{s-1}]\leq M_4<\infty
\]
almost surely. Define $\mathrm{PV}_t :=t^{-1}\sum_{s=1}^t \sigma^2_{s-1,1}(x_s)$, $\mathrm{MSE}_t :=t^{-1}\sum_{s=1}^t\bigl(y_s-\mu_{s-1,1}(x_s)\bigr)^2$, and $e_s:= \mu_{s-1,1}(x_s)-f(x_s)$. Assume that $e_s$ is $\mathcal F_{s-1}$-measurable and $|e_s|\leq E<\infty$ almost surely for all $s$, that $\widehat\sigma_t^2\xrightarrow{p}\sigma^2$, and that $\mathrm{PV}_t\xrightarrow{p}\mathrm{PV}_\infty$ for some finite constant $\mathrm{PV}_\infty\geq0$. Then
\begin{enumerate}[label=\roman*)]
    \item Well-specified: if $t^{-1} \sum_{s=1}^t e_s^2 \xrightarrow{p}0$, then $\widehat{\alpha}_t  \xrightarrow{p}1$.
    \item Misspecified: if $t^{-1} \sum_{s=1}^t e_s^2 \xrightarrow{p} b^2$ for some $0<b^2<\infty$, then
    \[
    \widehat{\alpha}_t \xrightarrow{p}  \sqrt{\frac{\mathrm{PV}_\infty+\sigma^2}{\mathrm{PV}_\infty+\sigma^2+b^2}}<1.
    \]
\end{enumerate}

\end{proposition}

\section{Experiments}
\label{sec:exps}
\subsection{Simulation study on benchmark optimization functions} \label{subsec:simulation}
We evaluate the effects of $\alpha$ and $g$ on a suite of challenging benchmark optimization functions \citep{simulationlib}. The function suite contains $47$ distinct functions exhibiting multiple local minima and diverse geometries. For functions without a fixed intrinsic dimension, we consider dimension $p\in\{5,10\}$; functions with a native dimension are used as is. This yields $61$ distinct function-dimension instances in total.

For each of the $61$ benchmark instances, we run $5$ independent Bayesian optimization trials (different random seeds) for every configuration of $(\alpha,g)$, with $\alpha\in\{1,\mathrm{Tempered}\}$ and $g\in\{0,1,2\}$. The GP posterior is adaptively tempered as described in Section \ref{sec:algorithmDesign}. For a $p$-dimensional objective function, we initialize with $n_0=\max\{5,2p\}$ observations and run the optimization for $T=10p$ additional evaluations.\footnote{We scale the evaluation budget linearly with dimension to compare the methods in a limited-evaluation regime. Substantially larger budgets for
low-dimensional functions tend to produce ceiling effects, as many methods eventually locate a global maximizer eventually.}  We consider a challenging high-noise regime $\sigma=2$. The corresponding low-noise results are reported in Appendix~\ref{subsec:appendix-simulation}. 

To assess robustness, we consider three additional noise settings beyond simple Gaussian noise $\epsilon \sim \mathcal{N}(0,\sigma^2)$, all normalized to have marginal standard deviation $\sigma$: 
\begin{enumerate}[label=(\alph*)]
    \item Gaussian mixture noise: $\epsilon \sim (1-\pi)\mathcal{N}(0,\widetilde{\sigma}^2) +\pi\mathcal{N}\bigl(0,(c\widetilde{\sigma})^2\bigr)$, where $\pi=0.1, c=5$, and $\widetilde{\sigma}=\frac{\sigma}{\sqrt{(1-\pi)+\pi c^2}}$ so that marginal standard deviation remains $\sigma$;
    \item Heteroscedastic noise: We allow the observations to be noisier near the optimum $x^\star$:
    $$\epsilon(x)\sim \mathcal{N}\left(0,\sigma^2\frac{\lambda(x)}{Z}\right), \qquad
      \lambda(x) = 1+\rho\exp\left\{ -\frac{\|x-x^\star\|^2}{\kappa p} \right\},$$
     where $\rho=\frac{1}{4}$, $\kappa=4$, and the normalizing constant $Z$ is estimated via Monte Carlo.
     \item Heavy-tailed Student's t-noise: $\epsilon = \sigma \sqrt{\frac{\nu-2}{\nu}}T_\nu$, with $\nu=4$.
\end{enumerate}
For each noise model, we conduct $61\times 5\times 2\times 3 = 1{,}830$ distinct BO runs.
The Gaussian experiment is closest to the theorem assumptions. We note that the Student-$t$ experiment is deliberately outside the conditional sub-Gaussian setting, and the adaptive $\widehat\alpha_t$ regime is not directly covered by the fixed-$\alpha$ regret theorem.

Table \ref{tab:alpha_noise_sigma2} summarizes pairwise comparisons of final simple regret. For each noise model and each value of $g$, final simple regret is averaged over the five random seeds for each benchmark instance.  The wins column gives the number of strict lower-regret outcomes for tempering versus $\alpha=1$, while the non-tie win rate conditions on instances where the two methods differ. Since exact ties can happen, the strict wins need not sum to $61$. The final two columns report paired improvements in average rank, oriented so that positive values favor tempering. The formal definitions of these metrics are given in Table \ref{tab:alpha_noise_sigma2}.

Tempering is favored in most regimes. Across the twelve noise-$g$ settings, it obtains more strict wins than $\alpha=1$ in nine settings and ties in one. The strongest gains occur at $g=1$: tempering wins $39$--$17$ under Gaussian noise and $36$--$18$ under heteroscedastic noise, with one-sided paired Wilcoxon tests giving $p=0.005$ and $p=0.013$ respectively. The Gaussian-mixture setting is also consistently favorable, with tempering winning all three values of $g$ and achieving non-tie win rates between $58.3\%$ and $62.7\%$.

The dependence on $g$ suggests that tempering is most effective in the intermediate regime represented by $g=1$. In the low-noise experiments reported in Appendix~\ref{subsec:appendix-simulation}, tempering is
especially helpful for $g=0$, where PI is highly exploitative and therefore vulnerable to premature posterior concentration. In the high-noise setting considered here, noisy comparisons with the incumbent may reduce the benefit of variance inflation for PI. At the other extreme, $g=2$ places greater weight on large improvements and already induces stronger exploration in many instances, leaving less room for tempering to alter the acquisition landscape. The strongest and most stable improvements therefore
occur at $g=1$, which balances the probability and magnitude of improvement. Across noise models, tempering is most consistently favorable under Gaussian and Gaussian-mixture noise; the heteroscedastic results depend more strongly on $g$, while the Student-$t$ results are mixed. This pattern is consistent with tempering as a correction for posterior overconfidence rather than as a substitute for a robust heavy-tailed likelihood.

\begin{table}[t]
\centering
\footnotesize
\setlength{\tabcolsep}{6pt}
\renewcommand{\arraystretch}{1.2}
\caption{Pairwise regret comparison of tempering against the untempered baseline
$\alpha=1$.}
\label{tab:alpha_noise_sigma2}
\begin{threeparttable}
\begin{tabular}{@{}l c c r r@{}}
\toprule
Noise & $g$ & Wins: Temp. vs. $\alpha=1$
& Non-tie win rate
& $\Delta$ rank \\
\midrule

Gaussian
  & 0 & $31$--$25$ & 55.4\% & $+0.10$ \\
  & 1 & $39$--$17$ & 69.6\% & $+0.36$ \\
  & 2 & $28$--$28$ & 50.0\% & $+0.00$ \\

\addlinespace[3pt]
Gaussian mixture
  & 0 & $35$--$25$ & 58.3\% & $+0.16$ \\
  & 1 & $37$--$22$ & 62.7\% & $+0.25$ \\
  & 2 & $34$--$24$ & 58.6\% & $+0.16$ \\

\addlinespace[3pt]
Heteroscedastic
  & 0 & $25$--$33$ & 43.1\% & $-0.13$ \\
  & 1 & $36$--$18$ & 66.7\% & $+0.30$ \\
  & 2 & $31$--$26$ & 54.4\% & $+0.08$ \\

\addlinespace[3pt]
Student-$t$
  & 0 & $31$--$26$ & 54.4\% & $+0.08$ \\
  & 1 & $28$--$31$ & 47.5\% & $-0.05$ \\
  & 2 & $31$--$27$ & 53.4\% & $+0.07$ \\

\bottomrule
\end{tabular}

\begin{tablenotes}[flushleft]
\scriptsize
\item Notes:
Entries compare the final simple regret over five seeds for each of the
$61$ benchmark instances. ``Wins: Temp. vs. $\alpha=1$'' reports the
numbers of strict lower-regret outcomes for the two methods, excluding ties.
The non-tie win rate is the fraction of non-tied comparisons won by
tempering. $\Delta$ rank is the average rank of $\alpha=1$ minus the
average rank of tempering, so positive values favor tempering.
\end{tablenotes}
\end{threeparttable}
\end{table}

To conclude the simulation study, Figure~\ref{fig:logregret_2x3} visualizes the regret trajectories for six benchmark instances under Gaussian noise, with $g\in\{0,1\}$. Each panel reports the median simple regret across the five random seeds, together with the interquartile range. The plots make clear that the gains from tempering are not merely terminal effects: in several cases, the tempered posterior drives regret down earlier and continues to attain lower regret throughout the search. This behavior is consistent with the interpretation that tempering can moderate posterior overconfidence, delay premature commitment to locally attractive regions, and preserve exploration during the early stages of
optimization.

\begin{figure}[t]
  \centering
  \includegraphics[width=0.75\linewidth, height=0.6\linewidth]{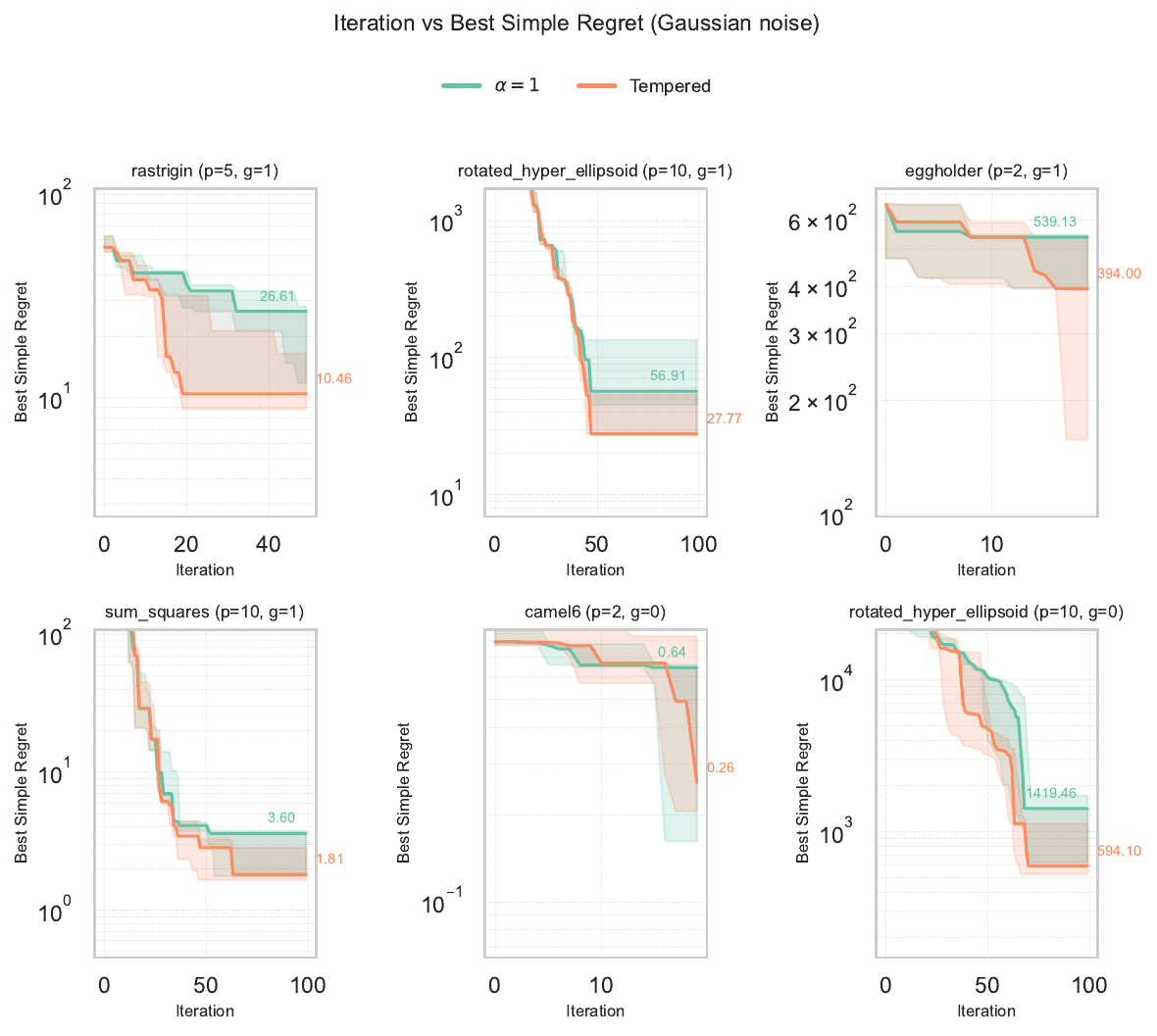}
  \caption{Effects of Tempering on Selected benchmarks, regret magnitude is shown on log scales.}
  \label{fig:logregret_2x3}
  \captionsetup{font=footnotesize}
  \vspace{2mm}
  {\footnotesize \emph{Notes.} Ribbons are the 25-75\% interquartile range across seeds at each iteration.  Note that for a $p$-dimensional function, we conducted $10p$ evaluations. }
\end{figure}

\subsection{Real data: materials optimization}

We use the \textit{Fe--Ga--Pd} materials-optimization dataset originally introduced by \citet{Long2007RapidMapping} and later adopted in the Bayesian active learning literature \citep{Kusne2020OnTheFly}. Each sample corresponds to a ternary alloy characterized by its composition fractions of iron (Fe), gallium (Ga), and palladium (Pd), which together satisfy the simplex constraint 
$\text{Fe} + \text{Ga} + \text{Pd} = 1, \quad \text{Fe}, \text{Ga}, \text{Pd} \ge 0$.
The response variable is the \textit{remanent magnetization}, measured as the output voltage of a scanning SQUID (Superconducting Quantum Interference Device) microscope \citep{Long2007RapidMapping}. Our objective is to identify the alloy composition that maximizes the SQUID voltage, corresponding to the highest remanent magnetization. Compositions with strong magnetic responses are of interest for applications in actuators and sensors. Because experimental evaluation of new compositions is both costly and time-consuming, Bayesian optimization provides a natural and efficient framework for this search.

Given the simplex constraint, we reparameterize the domain in two dimensions by $(\text{Fe}, \text{Ga})$, with $\text{Pd} = 1 - \text{Fe} - \text{Ga}$. For offline benchmarking, we fit a GP  model with a Matérn-$5/2$ kernel to all $278$ available measurements and treat its posterior mean function as the ground-truth black-box objective for subsequent experimentation.

To systematically examine the effects of posterior tempering accross different acquisition functions, we evaluate each method over real values of $g \in \{0, 0.5, 1, 1.5, 2\}$. \footnote{Because non-integer values of $g$ do not have easy to compute closed form solutions, we did not consider real-valued $g$ case in the large-scale simulation study in Section \ref{subsec:simulation}.} In additional to the $g$-EI family, we also include the Max-value Entropy Search (MES) acquisition function, 
which can be viewed as a variational approximation to MES \citep{cheng2025unified}. We consider the standard GP posterior $(\alpha=1)$. and the tempered posterior described in Section \ref{sec:algorithmDesign}. For each combination of $\alpha$ and acquisition function, we conduced $10$ independent random seeds of simulations, each with different initial sample size of 5.

Table \ref{tab:bo_g_alpha_voltage} compares the average best-observed SQUID voltages under different acquisition functions and tempering level $\alpha$. Tempering yields only marginal gains for MES, suggesting that when the acquisition function is already highly exploratory, additional variance tempering provides limited benefit. In contrast, tempering substantially improve the performance for the $g=0$ case, consistent with findings in Appendix \ref{subsec:low-noise} that tempered posteriors enhance probability-improvement behavior by promoting benign exploration. Notable gains are also observed for $g \in \{0.5, 1, 1.5 \}$, indicating that tempering remains advantageous when the acquisition balances exploration and exploitation. However, for $g=2$, tempering provides no improvement, reinforcing the pattern that excessive exploration diminishes the utility of posterior tempering.

\begin{table}[t]
\centering
\caption{Bayesian optimization performance by Acquisition Function and $\alpha$. 
Entries are average best observed SQUID voltage across all $10$ independent seeds at iterations 5-30.}
\label{tab:bo_g_alpha_voltage}
\begin{tabular}{
    l      
    c      
    S[table-format=1.2, round-mode=places, round-precision=2] 
    S[table-format=1.2, round-mode=places, round-precision=2] 
    S[table-format=1.2, round-mode=places, round-precision=2] 
    S[table-format=1.2, round-mode=places, round-precision=2] 
    S[table-format=1.2, round-mode=places, round-precision=2] 
    S[table-format=1.2, round-mode=places, round-precision=2] 
}
\toprule
\multicolumn{1}{c}{Acquisition $g$} & \multicolumn{1}{c}{$\alpha$} & 
\multicolumn{1}{c}{iter5} & \multicolumn{1}{c}{iter10} &
\multicolumn{1}{c}{iter15} & \multicolumn{1}{c}{iter20} &
\multicolumn{1}{c}{iter25} & \multicolumn{1}{c}{iter30} \\
\midrule
\multirow{2}{*}{MES} & 1  & \bfseries 5.244187 & \bfseries 5.542870 & \bfseries 5.687650 & 5.906776 & \bfseries 6.292336 & {\bfseries 6.54} \\
                     & Tempered & 5.047614 & 5.158649 & 5.551071 & \bfseries 6.017854 & 6.017854 & 6.105913 \\
\addlinespace[2pt]
\multirow{2}{*}{0.0} & 1  & \bfseries 5.072557 & 5.216988 & 5.335539 & 5.335539 & 5.389902 & 5.394254 \\
                     & Tempered & 4.985941 & \bfseries 5.509554 & \bfseries 5.985530 & \bfseries 5.990089 & \bfseries 6.256510 & \bfseries 6.256510 \\
\addlinespace[2pt]
\multirow{2}{*}{0.5} & 1  & \bfseries 5.435038 & \bfseries 6.089558 & 6.138694 & 6.646735 & 6.833154 & 6.834120 \\
                     & Tempered & 5.434079 & 5.984924 & \bfseries 6.998355 & \bfseries 7.017270 & \bfseries 7.018159 & \bfseries 7.047042 \\
\addlinespace[2pt]
\multirow{2}{*}{1.0} & 1  & \bfseries 5.290367 & \bfseries 5.439139 & 5.711982 & 5.793758 & 6.005960 & 6.208248 \\
                     & Tempered & 5.197762 & 5.202578 & \bfseries 5.852262 & \bfseries 6.858881 & \bfseries 7.192716 & \bfseries 7.264755 \\
\addlinespace[2pt]
\multirow{2}{*}{1.5} & 1  & 5.134462 &  5.399661 & 5.568679 & 5.683474 & 5.683833 & 5.806065 \\
                     & Tempered & 5.134462 & 5.399661 &  \bfseries 5.938927 & \bfseries 6.053723 &  \bfseries 6.053723 &  \bfseries 6.053723 \\
\addlinespace[2pt]
\multirow{2}{*}{2.0} & 1  & \bfseries 5.752573 & \bfseries 6.399660 & \bfseries 6.399660 & \bfseries 6.686206 & \bfseries 7.277116 & \bfseries 7.444661 \\
                     & Tempered & 5.511485 & 5.715830 & 5.765498 & 6.709797 & 6.709797 & 6.771566 \\
\bottomrule
\end{tabular}
\end{table}

\section{Discussion}\label{Discussion}
\noindent
We study Bayesian optimization with tempered posterior surrogates and generalized expected  improvement  acquisition functions, where the likelihood power $\alpha\in(0,1]$ regulates how strongly each new evaluation updates the surrogate and the exponent $g$ interpolates between probability of improvement and classic expected improvement. For a fixed $\alpha$, and a confidence-rescaled acquisition, our regret analysis makes the dependence on $(\alpha,g)$ explicit through information-gain quantities and exhibits a trade-off rather than a uniform benefit from tempering. Separately, we propose a sequential tuning-light schedule for choosing $\alpha_t$ and prove its calibration limit under the conditions of Proposition~\ref{prop:calibration2}. Empirically, we show that tempering tends to stabilize improvement-based BO in regimes where the surrogate can become overconfident under localized sampling, with the largest gains typically observed for not overly exploratory settings (e.g., PI and EI).

Our theoretical development focuses on generalized improvement-based acquisitions (the $g$-EI family) together with GP surrogates. There are are several potential extensions of our work. First, it would be of interest to understand whether likelihood tempering yields analogous benefits for other acquisition functions, such as knowledge gradient \citep{frazier2009knowledge} and entropy-based criteria (e.g., maximum-entropy search). Second, it remains open to derive regret guarantees for tempered updates when the surrogate departs from GP regression and scalar outputs, including tree-based  \citep{boyne2025bark,ohagan2024treebanditsgenerativebayes,luo2024hybrid}, neural-network surrogates \citep{srinivas2009gaussian} and tensor-functional outputs \citep{luo2025efficient}. Third, while our analysis treats $(\alpha,g)$ as fixed, developing methods that learn joint schedules $(\alpha_t,g_t)$ based on the current design with finite-time guarantees would strengthen the connection between acquisition-level exploration  and surrogate-level tempering trade-offs, even in BOI settings described in \citep{wang2025bayesian}. More broadly, the perspective of tempering as an online learning-rate mechanism may extend beyond BO to other sequential decision problems, including reinforcement learning \citep{Sutton1998}, dynamic treatment regimes \citep{murphy_2003}, and adaptive experimental design and testing \citep{10.1214/23-STS915,li2025deepcomputerizedadaptivetesting,srinivas2009gaussian}. Finally, it would be valuable to understand how tempering interacts with recent advances targeting challenging BO regimes, such as in high-dimensional domains, nonsmooth or nonstationary objectives, and mixed discrete-continuous search spaces, where complementary tools and benchmarks are actively being developed \citep{luo2022sparse,luo2024hybrid,cho2025surrogate,risser2024compactly,noack2025gp2scale}.

\section*{Disclosure statement}\label{disclosure-statement}
 HL was supported by U.S. Department of Energy under Contract DE-AC02-05CH11231 and U.S. National Science Foundation NSF-DMS 2412403. Data and code have been made available at the following URL: \url{https://github.com/JiguangLi/Bayesian-Optimization-via-Tempered-Posterior}.

\spacingset{1.2}
\bibliographystyle{plainnat}
\bibliography{vbo}
\clearpage
\newpage{}
\spacingset{1.2}
\appendix
\begin{center}
    \section*{SUPPLEMENTARY MATERIALS}
\end{center}
\section{Additional Motivation Results}
\label{sec:additional_toy_result}
\begin{table}[h]
\centering
\begin{tabular}{c|cccccc}
\hline
$\alpha$ & df = 1 & df = 2 & df = 5 & df = 10 & df = 20 & df = 100 \\
\hline
0.1 & \textbf{0.9963} & \textbf{0.9875} & \textbf{0.9856} & \textbf{0.9857} & \textbf{0.9857} & \textbf{0.9857} \\
0.5 & 0.9848 & 0.7644 & 0.7687 & 0.7608 & 0.7631 & 0.7631 \\
1.0 & 0.9665 & 0.7737 & 0.9842 & 0.9846 & 0.9848 & 0.9850 \\
\hline
\end{tabular}
\caption{\label{tab:example_PI}Best observed $y$ from PI runs for each $(\alpha, \text{df})$. Maximum per column is bold.}
\end{table}

\begin{table}[h]
\centering
\begin{tabular}{c|cccccc}
\hline
$\alpha$ & df = 1 & df = 2 & df = 5 & df = 10 & df = 20 & df = 100 \\
\hline
0.1 & 0.9787 & 0.9844 & 0.9919 & 0.9909 & 0.9904 & 0.9899 \\
0.5 & \textbf{1.3515} & \textbf{0.9901} & \textbf{1.0000} & \textbf{0.9957} & \textbf{0.9940} & \textbf{0.9925} \\
1.0 & 0.9875 & 0.9852 & 0.9883 & 0.9879 & 0.9879 & 0.9878 \\
\hline
\end{tabular}
\caption{\label{tab:example_EI}Best observed $y$ from EI runs for each $(\alpha, \text{df})$. Maximum per column is bold.}
\end{table} 

Figure \ref{fig:one-d-pi-tempered} in the main text presents a controlled one dimensional example in which three values of the tempering parameter \(\alpha \in \{0.1, 0.5, 1.0\}\) are used with probability of improvement. The figure shows that when the surrogate is fit with a tempered posterior with small $\alpha$, the posterior variance is effectively enlarged and the probability of improvement curve places mass in regions that have not yet been explored. As a result the policy that uses \(\alpha=0.1\) is able to locate a point with noticeably larger observed value than the policies that keep \(\alpha=0.5\) or \(\alpha=1.0\). We interpret  this as evidence that tempering can counter posterior over confidence and can restore exploratory moves when the acquisition is otherwise very exploitative.

In BO, model misspecification and the reuse of surrogate fits across iterations can make the GP posterior too sharp. A natural question is therefore whether the effect persists once we move away from the convenient Gaussian noise. The additional experiment keeps \eqref{eq:blackbox_fun}, the same acquisition, but it replaces the Gaussian noise by Student \(t\) noise with several degrees of freedom $\alpha$. 

Table \ref{tab:example_PI} reports the best value of the observed response that each of the three tempered policies attains, across degrees of freedom equal to \(1,2,5,10,20,100\). For every column the largest value is obtained by \(\alpha=0.1\). The advantage is most pronounced for degrees of freedom equal to \(1\) and \(2\), which correspond to the heaviest tails and therefore to the strongest mismatch. For larger degrees of freedom, the values for \(\alpha=0.1\) and \(\alpha=1.0\) become close, which is expected because the Student \(t\) noise approaches a Gaussian distribution and the original example already showed that the untempered policy can perform reasonably in that case.  

Table  \ref{tab:example_EI} reports the analogous experiment for expected improvement. In that case the maximum in each column is obtained for \(\alpha=0.5\). This agrees with the broader empirical conclusions of the paper, where it is noted that acquisitions that are already more balanced between exploitation and exploration benefit from a moderate amount of tempering but do not require the very small values of $\alpha$ that were effective for probability of improvement. 

This set of additional experiment shows that the main effect described in the text is acquisition dependent. For the strongly exploitative rule the smallest $\alpha$ is best, and for the more balanced rule an intermediate $\alpha$ is preferable. This is in harmony with the interpretation in the main text that tempered posteriors act as a control on the effective confidence of the surrogate.
\section{Connection to Stochastic Gradient Descent}\label{sec:motivations:equivalence}
This subsection presents an exact per iteration equality that links the $\alpha$ tempered Gaussian process update of the posterior mean with a single preconditioned stochastic gradient step on a one point loss. The statement holds for Gaussian observation noise and is conditional on the kernel fitted at the previous iteration, which is the convention adopted in sequential GP regression.

Consider an unknown target $f^\star$ and noisy observations $y_t=f^\star(x_t)+\varepsilon_t$ with $\varepsilon_t\sim\mathcal N(0,\sigma^2)$. Let the prior be $f\sim\mathcal{GP}(0,k)$ and after $t-1$ steps write the posterior as $\mathcal{GP}(\mu_{t-1},k_{t-1})$. For a candidate input location $x$ denote the predictive variance by $v_{t-1}(x):=k_{t-1}(x,x)$ and the cross covariance by $c_{t-1}(x,x_t):=k_{t-1}(x,x_t)$. At step $t$ temper the likelihood by a factor $\alpha_t>0$. Since the likelihood is Gaussian, tempering is equivalent to replacing the noise variance by $\sigma^2/\alpha_t$. The standard conditioning formula then yields the one point update of the mean
\begin{equation}\label{eq:bo_update_gain}
\mu_t(x)
\,=\,
\mu_{t-1}(x)
\,+\,
\frac{c_{t-1}(x,x_t)}{v_{t-1}(x_t)+\sigma^2/\alpha_t}
\left(y_t-\mu_{t-1}(x_t)\right).
\end{equation}
It is convenient to factor the gain as a scalar times a direction,
\begin{equation}\label{eq:eta_direction}
\frac{c_{t-1}(x,x_t)}{v_{t-1}(x_t)+\sigma^2/\alpha_t}
\,=\,
\eta_t(\alpha_t)\,c_{t-1}(x,x_t),
\qquad
\eta_t(\alpha_t)
\,=\,
\frac{\alpha_t}{\sigma^2+\alpha_t\,v_{t-1}(x_t)}.
\end{equation}
Substituting \eqref{eq:eta_direction} into \eqref{eq:bo_update_gain} gives the compact form
\begin{equation}\label{eq:bo_update_compact}
\mu_t
\,=\,
\mu_{t-1}
\,+\,
\eta_t(\alpha_t)\,
\left(y_t-\mu_{t-1}(x_t)\right)\,
c_{t-1}(\cdot,x_t).
\end{equation}

Now consider a functional optimization view at step $t$ based on the one point objective
\begin{equation}\label{eq:one_point_objective}
\mathcal J_t(g)
\,=\,
\frac{\alpha_t}{2\sigma^2}\left(y_t-g(x_t)\right)^2
\,+\,
\frac{1}{2}\,\|g-\mu_{t-1}\|_{\mathcal H_{k_{t-1}}}^{2},
\end{equation}
where $\mathcal H_{k_{t-1}}$ is the reproducing kernel Hilbert space associated with $k_{t-1}$. The negative log of the tempered posterior is equal to $\mathcal J_t$ up to an additive constant, so the minimizer of \eqref{eq:one_point_objective} coincides with the tempered posterior mean. By the representer theorem the minimizer has the form $g=\mu_{t-1}+\beta\,c_{t-1}(\cdot,x_t)$. Differentiation with respect to $\beta$ and evaluation at zero derivative give
\begin{equation}\label{eq:beta_star}
\beta^\star
\,=\,
\frac{\alpha_t}{\sigma^2+\alpha_t v_{t-1}(x_t)}\,
\left(y_t-\mu_{t-1}(x_t)\right)
\,=\,
\eta_t(\alpha_t)\,
\left(y_t-\mu_{t-1}(x_t)\right).
\end{equation}
Substituting $g=\mu_{t-1}+\beta^\star c_{t-1}(\cdot,x_t)$ leads exactly to the update in \eqref{eq:bo_update_compact}. Hence the tempered Gaussian update with $g$-EI equals a single preconditioned stochastic gradient step with direction $c_{t-1}(\cdot,x_t)$ and learning rate $\eta_t(\alpha_t)$ applied to the loss \eqref{eq:one_point_objective}.

We summarize the per iteration correspondence for a fixed datum $(x_t,y_t)$ for the mean function $\mu_t$:
\begin{align}\label{eq:equivalence_summary}
\text{tempered update with parameter }\alpha_t
\quad\Longleftrightarrow\quad\\
\text{stochastic gradient step with learning rate }\eta_t(\alpha_t)
\,=\,
\frac{\alpha_t}{\sigma^2+\alpha_t v_{t-1}(x_t)}.
\end{align}
When $\alpha_t$ is small the step is conservative and the mean update is close to the previous iterate, while large $\alpha_t$ values produce aggressive updates that are limited by the local variance $v_{t-1}(x_t)$. 

This identity is an equality of model updates that applies once the location $x_t$ and the observation $y_t$ are given, and shares a very similar form of robust tempering \citep{holmes2017assigning}. 
If a stochastic gradient method is supplied with the same sequence of data pairs then \eqref{eq:bo_update_compact} ensures that the two procedures produce the same mean path at the granularity of single iterations.

\section{Linear Surrogate BO Algorithm}

\begin{algorithm}[H]
\caption{\textsc{Tempered-EI Bayesian Optimization (Linear Surrogate)}}
\begin{algorithmic}[1]
\Require Compact domain $\mathcal{X}\subset\mathbb{R}^p$; feature map $\psi:\mathcal{X}\to\mathbb{R}^d$ with $\|\psi(x)\|_2\le L$; prior precision $\lambda>0$; noise std.\ $\sigma>0$; 
\Statex \hspace{1.6em} budget $T\in\mathbb{N}$; initial design $\{(x_s,y_s)\}_{s=1}^{t_0}$ (possibly $t_0=0$); tempering policy $\alpha_t\in(0,1]$
\Ensure Sequence $\{x_t\}_{t=1}^T$ and incumbent $\widehat{x}^*\in\arg\max_{x\in\mathcal{X}}\mu_{T,\alpha_T}(x)$
\State $t\gets t_0$; $\mathcal{D}_{t}\gets\{(x_s,y_s)\}_{s=1}^{t}$; $X_t\gets [\psi(x_1),\dots,\psi(x_t)]^\top$; $y_{1:t}\gets (y_1,\dots,y_t)^\top$
\State \textbf{while} $t<T$ \textbf{do}
\State \quad Pick tempering $\alpha_{t+1}\in(0,1]$ via a policy (e.g.\ fixed $\alpha$, or $\textsc{TemperSchedule}$)
\State \quad \textbf{Posterior update} w.r.t.\ $\alpha_{t+1}$ (Alg.~\ref{alg:tempered-posterior}): compute
\[
V_{t+1,\alpha}=\lambda I_d+\tfrac{\alpha_{t+1}}{\sigma^2}X_t^{T}X_t,\quad
\Sigma_{t+1,\alpha}=V_{t+1,\alpha}^{-1},\quad
\mu_{t+1,\alpha}=\Sigma_{t+1,\alpha}\tfrac{\alpha_{t+1}}{\sigma^2}X_t^{T}y_{1:t}.
\]
\State \quad Define predictive mean/variance for any $x\in\mathcal{X}$:
\[
\mu_{t+1,\alpha}(x)=\psi(x)^\top\mu_{t+1,\alpha},\qquad
\sigma^2_{t,\alpha_{t+1}}(x)=\psi(x)^\top\Sigma_{t+1,\alpha}\psi(x).
\]
\State \quad Compute current mean-maximum $m_t:=\displaystyle\max_{x\in\mathcal{X}}\mu_{t,\alpha_{t+1}}(x)$ \hfill\textit{(global or over a candidate set)}
\State \quad \textbf{Acquisition:} define $\mathrm{EI}_{t,\alpha_{t+1}}(x)$ via Alg.~\ref{alg:gei} $g=1$; pick
\[
x_{t+1}\in\arg\max_{x\in\mathcal{X}}\ \mathrm{EI}_{t,\alpha_{t+1}}(x).
\]
\State \quad Query black-box: $y_{t+1}\gets f(x_{t+1})+\epsilon_{t+1}$ with $\epsilon_{t+1}\sim\mathcal{N}(0,\sigma^2)$
\State \quad Augment data: $\mathcal{D}_{t+1}\gets \mathcal{D}_t\cup\{(x_{t+1},y_{t+1})\}$; $X_{t+1}\gets \begin{bmatrix}X_t \\ \psi(x_{t+1})^\top\end{bmatrix}$; $y_{1:t+1}\gets (y_{1:t},y_{t+1})$
\State \quad $t\gets t+1$
\State \textbf{end while}
\State \textbf{return} $\widehat{x}^*\in\arg\max_{x\in\mathcal{X}}\mu_{T,\alpha_T}(x)$ \textit{(or the best observed $y_t$-incumbent)}
\end{algorithmic}
\end{algorithm}

\section{Linear Surrogate for Tempered BO}
\label{sec:linear-ei-alpha}


Although our main focus is on Gaussian process (GP) surrogates with nonlinear kernels, we begin with Bayesian linear regression as a more analytically clean baseline. In this finite-dimensional setting, the $\alpha$-tempered posterior has a closed form and the effect of tempering can be traced explicitly through both the posterior mean and variance. The resulting regret bound shows that, under correct linear specification, likelihood tempering does not improve or jeopardize the leading-order worst-case guarantee for expected improvement (EI). This ``baseline non-improvement'' helps isolate where tempering can matter, motivating the GP analysis in Section~\ref{sec:GP-ei-alpha}, where kernel mismatch and adaptive localized sampling can induce local overconfidence.

Throughout this section we consider a compact decision set $\mathcal X \subset \mathbb R^{p}$ and a feature map $\psi:\mathcal X \to \mathbb R^{d}$ satisfying $\|\psi(x)\|_{2} \le L$ for all $x \in \mathcal X$. We model the latent reward as linear in the features,
$f(x) \;=\; \psi(x)^{\top} \theta^{\star}$,
where the true parameter $\theta^{\star} \in \mathbb R^{d}$ is bounded as $\|\theta^{\star}\|_{2} \le S_{\theta}$. 
Let $X_t = [\psi(x_1),\dots,\psi(x_t)]^{\top} \in \mathbb R^{t\times d}$ denote the design matrix and $y_{1:t} = (y_1,\dots,y_t)^{\top}$ the vector of observations. We place a Gaussian prior $\theta \sim \mathcal N(0,\lambda^{-1} I_d)$ with $\lambda>0$ and introduce a tempering parameter $\alpha \in (0,1]$. 
Since $y_{1:t}\mid\theta\sim\mathcal{N}(X_t\theta,\sigma^2 I_t)$ and $\theta\sim\mathcal{N}(0,\lambda^{-1}I_d)$, the $\alpha$-tempered posterior according to \eqref{eq:alpha_post} is
\begin{align*}
\log p_\alpha(\theta\mid\mathcal{D}_t)
&=\mathrm{const}-\frac{1}{2}\left[\lambda\|\theta\|_2^2+\frac{\alpha}{\sigma^2}\|y_{1:t}-X_t\theta\|_2^2\right]\\
&=\mathrm{const}-\frac{1}{2}\Big[\theta^{T}V_{t,\alpha}\theta-2\,\theta^{T}\tfrac{\alpha}{\sigma^2}X_t^{T}y_{1:t}\Big].
\end{align*}
Completing the square gives
$\theta\mid\mathcal{D}_t,\alpha\sim\mathcal{N}\!\big(\mu_{t,\alpha},\,\Sigma_{t,\alpha}\big)$, where
\[
V_{t,\alpha} \;=\; \lambda I_d + \frac{\alpha}{\sigma^{2}} X_t^{\top} X_t,
\qquad
\Sigma_{t,\alpha} \;=\; V_{t,\alpha}^{-1},
\qquad \mu_{t,\alpha}
=\Sigma_{t,\alpha}\,\frac{\alpha}{\sigma^{2}} X_t^{\top} y_{1:t}.
\]
For any $x \in \mathcal X$ the tempered predictive mean and variance of the latent reward are $\mu_{t,\alpha}(x; \theta) = \psi(x)^{\top} \mu_{t,\alpha}$, and
$\sigma^2=\sigma_{t,\alpha}^{2}(x; \theta)= \psi(x)^{\top} \Sigma_{t,\alpha} \psi(x)$.

\subsection{Regret Analysis of Tempered Linear Surrogate}
\label{subsec:linear_fractional}
For a maximizer $x^{\star} \in \arg \max_{x \in \mathcal{X}} f(x)$, the instantaneous and cumulative regrets are: 
\begin{equation} \label{eq:regret_def}
r_t \;=\; f(x^{\star}) - f(x_t),
\qquad
R_T \;=\; \sum_{t=1}^{T} r_t.
\end{equation}
In the remaining part of Section \ref{sec:linear-ei-alpha}, we focus on expected improvement ($g=1$) and write $\mu_{t,\alpha}(x):=\mu_{t,\alpha}(x;\theta)$ and $\sigma_{t,\alpha}(x):=\sigma_{t,\alpha}(x;\theta)$ for convenience. We assume that 
\begin{assumption}
    \label{ass:radius}
    There exists $\theta^\star\in\mathbb{R}^d$ such that
$f(x)=\psi(x)^\top\theta^\star$ for all $x\in\mathcal{X}$ and
$\|\theta^\star\|_2\le S_\theta$ for a known constant $S_\theta>0$.
\end{assumption}

Write $m_{t-1}:=\mu_{\theta_{t-1},\alpha}^{+}= \max_{x \in \mathcal X} \mu_{t-1,\alpha}(x; \theta)$
as the running posterior mean maximum. At round $t$, the tempered predictive distribution follows $\mathcal{N}\big(\mu_{t-1,\alpha}(x),\sigma^2_{t-1,\alpha}(x)\big)$, and the expected improvement $\mathrm{EI}_{t-1,\alpha}(x)
:=\mathbb{E}\!\left[(f(x)-m_{t-1})_+\mid\mathcal{D}_{t-1}\right]$ can be written as
\begin{align*}
\mathrm{EI}_{t-1,\alpha}(x)
&=(\mu_{t-1,\alpha}(x)-m_{t-1})\,\Phi(z_{t-1}(x))
+\sigma_{t-1,\alpha}(x)\,\phi(z_{t-1}(x)),\\
z_{t-1}(x)&:=\frac{\mu_{t-1,\alpha}(x)-m_{t-1}}{\sigma_{t-1,\alpha}(x)}.
\end{align*}
The EI policy selects any maximizer
$x_t\in\arg\max_{x\in\mathcal{X}} \ \mathrm{EI}_{t-1,\alpha}(x)$

Our regret argument uses an \emph{EI-UCB alignment} condition below: it ensures that, over the attainable $(\mu,\sigma)$ pairs at round $t$, the ordering induced by EI is consistent with the linear scalarization $\mu+\kappa_t\sigma$. This is a rather strong technical condition (See Remark 8 of \citet{wang2025bayesian}); we impose it only in this warm-up for our intuition and 
in Section \ref{sec:GP-ei-alpha} we analyze generalized improvement directly and do not rely on such an alignment assumption.

\begin{assumption}
\label{ass:alignment}
There exists $\kappa_t\ge 0$ such that for all $x_1,x_2\in\mathcal{X}$,
\[
\mu_{t-1,\alpha}(x_1)+\kappa_t\,\sigma_{t-1,\alpha}(x_1)\ >\ \mu_{t-1,\alpha}(x_2)+\kappa_t\,\sigma_{t-1,\alpha}(x_2)
\ \Longrightarrow\
\mathrm{EI}_{t-1,\alpha}(x_1)\ >\ \mathrm{EI}_{t-1,\alpha}(x_2).
\]
Thus the alignment condition preserves strict score order; no conclusion is required when the two linear scores tie.
\end{assumption}

\begin{theorem}[EI regret bound (fixed $\alpha$)]
\label{thm:ei-regret}
Suppose Assumption \ref{ass:radius}-\ref{ass:alignment} holds at each round $t$ with:
\[
\kappa_t:=\beta_{t-1}(\alpha,\delta)=\sqrt{\lambda}S_\theta+\sqrt{\alpha}
\sqrt{\,\log\!\frac{\det V_{t-1,\alpha}}{\det(\lambda I_d)}+2\log\!\frac{1}{\delta}\,}.
\]
Then, with probability at least $1-\delta$, for all $T\ge 1$,
\begin{align}\label{eq:EI_regret}
R_T\ \le\ 2\,\beta_T(\alpha,\delta)\,
\sqrt{\,c_{\alpha,\lambda}\,T\cdot \log\!\frac{\det V_{T,\alpha}}{\det(\lambda I_d)}\,},
\qquad
c_{\alpha,\lambda}:=\frac{2\sigma^2}{\alpha}+\frac{L^2/\lambda}{\log 2}.
\end{align}
\end{theorem}
To control the term $\log\!\frac{\det V_{T,\alpha}}{\det(\lambda I_d)}$ in $R_T$ of \eqref{eq:EI_regret}, we will need the following Lemma.

\begin{lemma}[Determinant growth]
\label{lem:det-growth}
For all $T\ge 1$,
\[
\log\!\frac{\det V_{T,\alpha}}{\det(\lambda I_d)}
\ \le\ d\,\log\!\Big(1+\frac{\alpha L^2 T}{\lambda\sigma^2\,d}\Big).
\]
\end{lemma}

\begin{corollary}
\label{cor:ei-explicit}
Combining with Lemma~\ref{lem:det-growth},
\[
R_T\ \le\ 2\,\beta_T(\alpha,\delta)\,
\sqrt{\,c_{\alpha,\lambda}\,d\,T\,\log\!\Big(1+\frac{\alpha L^2 T}{\lambda\sigma^2\,d}\Big)\,},
\]
where
\[
\beta_T(\alpha,\delta)=\sqrt{\lambda}S_\theta+ \sqrt{\alpha}
\sqrt{\,d\log\!\Big(1+\frac{\alpha L^2 T}{\lambda\sigma^2\,d}\Big)+2\log\!\frac{1}{\delta}\,},
\qquad
c_{\alpha,\lambda}=\frac{2\sigma^2}{\alpha}+\frac{L^2/\lambda}{\log 2}.
\]
\end{corollary}

Corollary \ref{cor:ei-explicit} makes the $\alpha$-dependence explicit. The factor $\sqrt{\alpha}$ in $\beta_T(\alpha, \delta)$ cancels the $\sqrt{1/\alpha}$ factor in $c_{\alpha,\lambda}$, leaving only a slowly varying logarithmic term. On the other hand, The prior component $$\sqrt{\lambda} S_{\theta} \sqrt{\, (\frac{2\sigma^2}{\alpha}+\frac{L^2/\lambda}{\log 2}) \,d\,T\,\log\!\Big(1+\frac{\alpha L^2 T}{\lambda\sigma^2\,d}\Big)\,}$$ typically favors larger $\alpha$ (closer to $1$), up to logarithmic factors. Thus, in the correctly specified linear model and under Assumption \ref{ass:radius}-\ref{ass:alignment}, tempering is not expected to improve worst-case cumulative regret for EI, and may slightly worsen constants.


This negative result is informative rather than discouraging: it suggests that any theoretical or empirical gains from tempering should be attributed to settings where the surrogate can become locally overconfident under adaptive design, a phenomenon that is substantially more plausible for nonlinear GP kernels. We turn to this regime in Section \ref{sec:GP-ei-alpha}, where we analyze tempered GP surrogates and generalized improvement acquisitions without requiring an EI-UCB alignment condition. We move beyond this linear parametric kernel to Mat\'ern and squared-exponential kernels, where the induced feature map is effectively infinite-dimensional and the interaction between acquisition-level choices and surrogate-level tempering, becomes nontrivial.


\begin{algorithm}[H]
\caption{\textsc{Tempered-EI Bayesian Optimization (Linear Surrogate)}}
\label{alg:tempered-ei-bo}
\begin{algorithmic}[1]
\Require Compact domain $\mathcal{X}\subset\mathbb{R}^p$; feature map $\psi:\mathcal{X}\to\mathbb{R}^d$ with $\|\psi(x)\|_2\le L$; prior precision $\lambda>0$; noise std.\ $\sigma>0$; 
\Statex \hspace{1.6em} budget $T\in\mathbb{N}$; initial design $\{(x_s,y_s)\}_{s=1}^{t_0}$ (possibly $t_0=0$); tempering policy $\alpha_t\in(0,1]$
\Ensure Sequence $\{x_t\}_{t=1}^T$ and incumbent $\widehat{x}^*\in\arg\max_{x\in\mathcal{X}}\mu_{T,\alpha_T}(x)$
\State $t\gets t_0$; $\mathcal{D}_{t}\gets\{(x_s,y_s)\}_{s=1}^{t}$; $X_t\gets [\psi(x_1),\dots,\psi(x_t)]^\top$; $y_{1:t}\gets (y_1,\dots,y_t)^\top$
\State \textbf{while} $t<T$ \textbf{do}
\State \quad Pick tempering $\alpha_{t+1}\in(0,1]$ via a policy (e.g.\ fixed $\alpha$, or $\textsc{TemperSchedule}$)
\State \quad \textbf{Posterior update} w.r.t.\ $\alpha_{t+1}$ (Alg.~\ref{alg:tempered-posterior}): compute
\[
V_{t+1,\alpha}=\lambda I_d+\tfrac{\alpha_{t+1}}{\sigma^2}X_t^{T}X_t,\quad
\Sigma_{t+1,\alpha}=V_{t+1,\alpha}^{-1},\quad
\mu_{t+1,\alpha}=\Sigma_{t+1,\alpha}\tfrac{\alpha_{t+1}}{\sigma^2}X_t^{T}y_{1:t}.
\]
\State \quad Define predictive mean/variance for any $x\in\mathcal{X}$:
\[
\mu_{t+1,\alpha}(x)=\psi(x)^\top\mu_{t+1,\alpha},\qquad
\sigma^2_{t,\alpha_{t+1}}(x)=\psi(x)^\top\Sigma_{t+1,\alpha}\psi(x).
\]
\State \quad Compute current mean-maximum $m_t:=\displaystyle\max_{x\in\mathcal{X}}\mu_{t,\alpha_{t+1}}(x)$ \hfill\textit{(global or over a candidate set)}
\State \quad \textbf{Acquisition:} define $\mathrm{EI}_{t,\alpha_{t+1}}(x)$ via Alg.~\ref{alg:gei} $g=1$; pick
\[
x_{t+1}\in\arg\max_{x\in\mathcal{X}}\ \mathrm{EI}_{t,\alpha_{t+1}}(x).
\]
\State \quad Query black-box: $y_{t+1}\gets f(x_{t+1})+\epsilon_{t+1}$ with $\epsilon_{t+1}\sim\mathcal{N}(0,\sigma^2)$
\State \quad Augment data: $\mathcal{D}_{t+1}\gets \mathcal{D}_t\cup\{(x_{t+1},y_{t+1})\}$; $X_{t+1}\gets \begin{bmatrix}X_t \\ \psi(x_{t+1})^\top\end{bmatrix}$; $y_{1:t+1}\gets (y_{1:t},y_{t+1})$
\State \quad $t\gets t+1$
\State \textbf{end while}
\State \textbf{return} $\widehat{x}^*\in\arg\max_{x\in\mathcal{X}}\mu_{T,\alpha_T}(x)$ \textit{(or the best observed $y_t$-incumbent)}
\end{algorithmic}
\end{algorithm}

\subsection{Proof of Lemma \ref{lem:det-growth}}
\begin{proof}
By Sylvester’s theorem,
\(
\frac{\det V_{T,\alpha}}{\det(\lambda I_d)}
=\det\!\big(I_d+\frac{\alpha}{\lambda\sigma^2}X_T^{T}X_T\big)
=\prod_{j=1}^d (1+\eta_j),
\)
where $\eta_j\ge 0$ are eigenvalues of $\frac{\alpha}{\lambda\sigma^2}X_T^{T}X_T$ and
$\sum_j\eta_j=\frac{\alpha}{\lambda\sigma^2}\operatorname{tr}(X_T^{T}X_T)\le \frac{\alpha L^2 T}{\lambda\sigma^2}$.
Concavity of $\log(1+x)$ and Jensen give the bound.
\end{proof}

\subsection{UCB Argument}
Due to Assumption \ref{ass:radius}, it suffices to study the regret bound of upper confidence bound acquisition by Lemma \ref{lem:argmax-eq}. 
Define the filtration $\mathcal{F}_t:=\sigma(x_1,\epsilon_1,\dots,x_t,\epsilon_t)$, the standardized noise
\[
\xi_t:=\epsilon_t/\sigma\sim\mathcal{N}(0,1)\ \text{i.i.d.},
\]
and the rescaled features $z_t:=\frac{\sqrt{\alpha}}{\sigma}\,\psi(x_t)\in\mathbb{R}^d$. Let $Z_t=[z_1,\dots,z_t]^\top$ and
\[
W_t:=\sum_{i=1}^t z_i\,\xi_i=Z_t^{T}\xi_{1:t}\in\mathbb{R}^d,
\qquad V_{t,\alpha}=\lambda I_d+Z_t^{T}Z_t.
\]
We begin by introducing two auxilliary lemmas:
\begin{lemma}[Coordinate-wise monotonicity of EI]
\label{lem:ei-monotone}
Let $\mathrm{EI}(\mu,\sigma;m)=(\mu-m)\Phi(\frac{\mu-m}{\sigma})+\sigma\phi(\frac{\mu-m}{\sigma})$. For all $\mu\in\mathbb{R}$, $\sigma>0$,
\[
\frac{\partial}{\partial \mu}\mathrm{EI}(\mu,\sigma;m)=\Phi\!\left(\frac{\mu-m}{\sigma}\right)\ge 0,\qquad
\frac{\partial}{\partial \sigma}\mathrm{EI}(\mu,\sigma;m)=\phi\!\left(\frac{\mu-m}{\sigma}\right)\ge 0.
\]
\end{lemma}

\begin{proof}
Differentiate the closed form using $\frac{\partial}{\partial \mu}(\mu-m)=1$,
$\frac{\partial}{\partial \mu}\big(\frac{\mu-m}{\sigma}\big)=\frac{1}{\sigma}$,
$\frac{\partial}{\partial \sigma}\big(\frac{\mu-m}{\sigma}\big)=-\frac{\mu-m}{\sigma^2}$,
$\Phi'=\phi$ and $\phi'=-z\phi$.
\end{proof}
 For example, if the set
$\{(\mu_{t-1,\alpha}(x),\sigma_{t-1,\alpha}(x)):x\in\mathcal{X}\}$ inside $\{\mu\le m_{t-1}\}$ lies on a nonincreasing curve $h(.)$ such that
$\sigma=h(\mu)$, then by Lemma~\ref{lem:ei-monotone} the order induced by any increasing function of $(\mu,\sigma)$ (such as EI or $\mu+\kappa\sigma$ with $\kappa\ge 0$) is the same along that curve.

\begin{lemma}[Argmax equivalence under alignment]
\label{lem:argmax-eq}
If Assumption~\ref{ass:alignment} holds, then every EI maximizer is also a maximizer of the linear score
\[
\mathrm{sc}_{t-1,\kappa_t}(x):=\mu_{t-1,\alpha}(x)+\kappa_t\,\sigma_{t-1,\alpha}(x).
\]
\end{lemma}

\begin{proof}
Let $x_t\in\arg\max \mathrm{EI}_{t-1,\alpha}$. If some $x$ satisfies $\mathrm{sc}_{t-1,\kappa_t}(x)>\mathrm{sc}_{t-1,\kappa_t}(x_t)$ with both means $\le m_{t-1}$, Assumption~\ref{ass:alignment} implies $\mathrm{EI}_{t-1,\alpha}(x)>\mathrm{EI}_{t-1,\alpha}(x_t)$, a contradiction.
\end{proof}

\begin{lemma}[Self-normalized bound]
\label{lem:selfnorm}
For $\delta\in(0,1)$, with probability $\ge 1-\delta$, simultaneously for all $t\ge 0$,
\[
\|W_t\|_{V_{t,\alpha}^{-1}}^2\ \le\ 2\log\!\left(\frac{\det(V_{t,\alpha})^{1/2}}{\det(\lambda I_d)^{1/2}\,\delta}\right).
\]
\end{lemma}

\begin{proof}
Let $A_{t,\alpha}:=\sum_{i=1}^t z_i z_i^{T}=V_{t,\alpha}-\lambda I_d$ and consider
$L_t(\theta):=\exp(\theta^{T}W_t-\tfrac12\|\theta\|_{A_{t,\alpha}}^2)$. As in standard proofs,
$\{L_t(\theta)\}$ is a martingale for each fixed $\theta\in\mathbb{R}^d$ because conditionally on $\mathcal{F}_{t-1}$,
$W_t=W_{t-1}+z_t\xi_t$, $A_{t,\alpha}=A_{t-1,\alpha}+z_tz_t^T$, and $\mathbb{E}[e^{(\theta^{T}z_t)\xi_t-\frac12(\theta^{T}z_t)^2}]=1$ for $\xi_t\sim\mathcal{N}(0,1)$.
Mix with $\Theta\sim\mathcal{N}(0,\lambda^{-1}I_d)$ (independent) to obtain a nonnegative martingale
\[
M_t:=\frac{\det(\lambda I_d)^{1/2}}{\det(V_{t,\alpha})^{1/2}}\,
\exp\!\Big(\tfrac{1}{2}\|W_t\|_{V_{t,\alpha}^{-1}}^2\Big),
\]
with $\mathbb{E}M_t=1$. Ville’s inequality gives
\[
\mathbb{P}\!\left(\exists t:\ \tfrac12\|W_t\|_{V_{t,\alpha}^{-1}}^2
\ge \log\frac{\det(V_{t,\alpha})^{1/2}}{\det(\lambda I_d)^{1/2}\,\delta}\right)\le \delta,
\]
which is the claim.
\end{proof}

\begin{lemma}[Uniform linear $\alpha$-confidence]
\label{lem:alpha-confidence}
With probability $\ge 1-\delta$, for all $t\ge 0$ and $x\in\mathcal{X}$,
\[
|f(x)-\mu_{t,\alpha}(x)|\ \le\ \beta_t(\alpha,\delta)\,\sigma_{t,\alpha}(x),
\quad
\beta_t(\alpha,\delta):=\sqrt{\lambda}\,S_\theta+\sqrt{\alpha}
\sqrt{\,\log\!\frac{\det V_{t,\alpha}}{\det(\lambda I_d)}+2\log\!\frac{1}{\delta}\,}.
\]
\end{lemma}

\begin{proof}
From $\mu_{t,\alpha}-\theta^\star=\Sigma_{t,\alpha}\frac{\alpha}{\sigma^2}X_t^\top\epsilon_{1:t}-\lambda\Sigma_{t,\alpha}\theta^\star$,
\[
f(x)-\mu_{t,\alpha}(x)=\underbrace{\psi(x)^\top\lambda\Sigma_{t,\alpha}\theta^\star}_{A}
-\underbrace{\psi(x)^\top\Sigma_{t,\alpha}\tfrac{\alpha}{\sigma^2}X_t^\top\epsilon_{1:t}}_{B}.
\]
We first apply Cauchy-Schwarz to $A$ in the $\|.\|_{\Sigma_{t,\alpha}}$ norm. Since $\theta^T\Sigma_{t,\alpha} \theta \leq \theta^T(\lambda^{-1}\mathbb{I}_d)\theta$, this yields
$$A \leq \|\psi(x)\|_{\Sigma_{t,\alpha}} \|\lambda \theta^*\|_{{\Sigma_{t,\alpha}}} \leq \sigma_{t,\alpha}(x) \sqrt{\lambda} \|\theta^*\|_2 \leq \sigma_{t,\alpha}(x) \sqrt{\lambda} S_{\theta}.$$
We again apply Cauchy-Schwarz to $B$ in the $\|.\|_{\Sigma_{t,\alpha}}$ norm. This yields
$$B \leq \|\psi(x)\|_{\Sigma_{t,\alpha}} \|\sqrt{\alpha} W_t\|_{\Sigma_{t,\alpha}} \leq \sigma_{t,\alpha}(x) \sqrt{\alpha} \|W_t\|_{V_{t,\alpha}^{-1}}.$$
Apply Lemma~\ref{lem:selfnorm} onto part B and multiply by $\sigma/\sqrt{\alpha}$ will yield this bound.
\end{proof}

\subsection{Proof of the Theorem \ref{thm:ei-regret}}

\begin{lemma}
\label{lem:det-recursion}
Let $u_t:=\psi(x_t)^{T}V_{t-1,\alpha}^{-1}\psi(x_t)$. Then
\[
\frac{\det V_{t,\alpha}}{\det V_{t-1,\alpha}}=1+\frac{\alpha}{\sigma^2}\,u_t,\qquad
\log\!\frac{\det V_{T,\alpha}}{\det(\lambda I_d)}=\sum_{t=1}^T \log\!\Big(1+\frac{\alpha}{\sigma^2}\,u_t\Big).
\]
\end{lemma}
\begin{proof} This is the matrix determinant lemma: $\det(A+ab^\top)=\det(A)\,(1+b^{T}A^{-1}a)$, applied to $A=V_{t-1,\alpha}$ and $a=b=\sqrt{\alpha}\,\psi(x_t)/\sigma$. \end{proof}

\begin{lemma}[Accumulated posterior variance]
\label{lem:sum-variance}
Let $u_t:=\psi(x_t)^{T}V_{t-1,\alpha}^{-1}\psi(x_t)$, then With $z_t:=\frac{\alpha}{\sigma^2}u_t\ge 0$,
\[
\sum_{t=1}^T u_t\ \le\
\left(\frac{2\sigma^2}{\alpha}+\frac{L^2/\lambda}{\log 2}\right)\,
\log\!\frac{\det V_{T,\alpha}}{\det(\lambda I_d)}.
\]
\end{lemma}

\begin{proof}
Split indices into $\mathcal{S}=\{t:\ z_t\le 1\}$ and $\mathcal{L}=\{t:\ z_t>1\}$. If $z_t\le 1$, then
$z_t\le 2\log(1+z_t)$, hence
$u_t=\frac{\sigma^2}{\alpha}z_t\le \frac{2\sigma^2}{\alpha}\log(1+z_t)$.
If $z_t>1$, then $\log(1+z_t)\ge \log 2$ and
$u_t\le \psi(x_t)^\top(\lambda I_d)^{-1}\psi(x_t)\le L^2/\lambda$,
so $\sum_{t\in\mathcal{L}}u_t\le \frac{L^2/\lambda}{\log 2}\sum_{t\in\mathcal{L}}\log(1+z_t)$.
Sum both parts and use Lemma~\ref{lem:det-recursion} to yield the result.
\end{proof}

\begin{proof}
Working on the $1-\delta$ event of $|f(x)-\mu_{t,\alpha}(x)|\ \le\ \beta_t(\alpha,\delta)\,\sigma_{t,\alpha}(x),$ (See Lemma~\ref{lem:alpha-confidence}), for each $t$ and all $x$,
\[
f(x)\le \mu_{t-1,\alpha}(x)+\kappa_t\,\sigma_{t-1,\alpha}(x),\qquad
f(x)\ge \mu_{t-1,\alpha}(x)-\kappa_t\,\sigma_{t-1,\alpha}(x).
\]
Thus
\[
r_t=f(x^\star)-f(x_t)
\le \big[\mu_{t-1,\alpha}(x^\star)+\kappa_t\sigma_{t-1,\alpha}(x^\star)\big]
-\big[\mu_{t-1,\alpha}(x_t)-\kappa_t\sigma_{t-1,\alpha}(x_t)\big].
\]
By Lemma~\ref{lem:argmax-eq}, the EI maximizer $x_t$ maximizes $\mathrm{sc}_{t-1,\kappa_t}$, hence
\(\mu_{t-1,\alpha}(x^\star)+\kappa_t\sigma_{t-1,\alpha}(x^\star)\le \mu_{t-1,\alpha}(x_t)+\kappa_t\sigma_{t-1,\alpha}(x_t)\),
and therefore
\(
r_t\le 2\kappa_t\,\sigma_{t-1,\alpha}(x_t)\le 2\beta_T(\alpha,\delta)\,\sigma_{t-1,\alpha}(x_t)
\)
since $\beta_{t-1}\le \beta_T$. Sum over $t$, use Cauchy-Schwarz and Lemma~\ref{lem:sum-variance} to yield the result.
\end{proof}

\section{GP Surrogate BO Algorithm} \label{sec:algo-pseudo}
Special cases: $g{=}0$ gives probability of improvement $\Phi(-v)$; $g{=}1$ recovers classical EI: $\mathrm{gEI}^{(1)}=\sigma\big(\phi(v)-v\,\Phi(-v)\big)$.

\begin{algorithm}[H]
\caption{\textsc{BO-$\alpha$GP-gEI}: Bayesian Optimization via $\alpha$-Tempered GP Posterior and Generalized EI}
\label{alg:bo-alpha-gp-gei}
\begin{algorithmic}[1]
\Require Domain $\X\subset\RR^p$; GP prior $(m,k_\theta)$; noise variance $\sigma^2>0$ (or estimator); budget $T$.
\Require Generalized-EI order $g\in\{0,1,2,\dots\}$ (fixed), or a scheduler $\textsc{gSchedule}$ returning $g_t$.
\Require Tempering scheduler $\textsc{TemperSchedule}$ returning $\alpha_t\in(0,1]$ (e.g., fixed or data-driven); acquisition rescaling $\nu_t>0$; optional jitter $\xi_t\ge0$.
\Require (Optional) hyperparameter learning routine $\textsc{FitHyperparams}$ (e.g., marginal likelihood).
\State Initialize design $\{x_1,\dots,x_{t_0}\}\subset\X$ (space-filling or random); evaluate $y_s=f(x_s)+\epsilon_s$, $s=1,\dots,t_0$.
\State Set $t\gets t_0$; $\mathcal{D}_t\gets\{(x_s,y_s)\}_{s=1}^{t}$.
\For{$t=t_0, t_0{+}1, \dots, T{-}1$}
  \State $\alpha_t \gets \textsc{TemperSchedule}(\mathcal{D}_t)$ \Comment{$\alpha_t\in(0,1]$; e.g., $\alpha_t\equiv 1$}
  \State $(\theta,m)\gets \textsc{FitHyperparams}(\mathcal{D}_t)$ \Comment{Optional re-fit kernel/mean}
  \State $(\mu_{t,\alpha_t},\sigma_{t,\alpha_t}) \gets \textsc{TemperedPosteriorGP}(\mathcal{D}_t, m, k_\theta, \sigma^2, \alpha_t)$
  \State $\mu^+_{t,\alpha_t} \gets \max_{x\in\X}\ \mu_{t,\alpha_t}(x)$ \Comment{Global maximization or dense grid}
  \State $g_t \gets \textsc{gSchedule}(t,\mathcal{D}_t)$ \textbf{or} $g_t\gets g$ \Comment{Fixed or adaptive $g$}
  \State Define acquisition
        \[
          a_t(x) \gets \textsc{gEI}\big(\mu_{t,\alpha_t}(x),\ \nu_t\sigma_{t,\alpha_t}(x),\ \mu^+_{t,\alpha_t}+\xi_t,\ g_t\big).
        \]
  \State $x_{t+1} \gets \arg\max_{x\in\X} a_t(x)$ \Comment{Global maximization of g-EI}
  \State Query black-box: $y_{t+1} \gets f(x_{t+1}) + \epsilon_{t+1}$, with $\epsilon_{t+1}\sim\mathcal{N}(0,\sigma^2)$
  \State Augment data: $\mathcal{D}_{t+1} \gets \mathcal{D}_t \cup \{(x_{t+1},y_{t+1})\}$
\EndFor
\State \Return $x_{\mathrm{rec}} \in \arg\max_{s\in\{1,\dots,T\}} \mu_{T,\alpha_T}(x_s)$ \Comment{Or best observed $y_s$}
\end{algorithmic}
\end{algorithm}

\noindent\textit{Theory/practice modes.} Theorem~\ref{thm:general-g-alpha} corresponds to fixed $\alpha_t\equiv\alpha$, predictable hyperparameters $\theta_t$ satisfying Assumption~\ref{ass:uniform-theta}, zero jitter, and $\nu_t$ satisfying Assumption~\ref{ass:nu-scaling}. The experiments use $\nu_t=1$, adaptive $\widehat\alpha_t$, and optional hyperparameter refitting; these choices are not covered by the finite-time theorem. For $g=0$ and $\xi_t=0$, the acquisition is posterior-mean greedy, whereas $\xi_t>0$ gives nondegenerate PI.


\begin{algorithm}[H]
\caption{\textsc{TemperedPosteriorGP}$(\mathcal{D}_t, m, k_\theta, \sigma^2, \alpha)$}
\label{alg:tempered-posterior}
\begin{algorithmic}[1]
\Require $\mathcal{D}_t=\{(x_s,y_s)\}_{s=1}^t$, mean $m$, kernel $k_\theta$, noise $\sigma^2$, tempering $\alpha\in(0,1]$.
\State Build $K_t$ with $(K_t)_{ij}=k_\theta(x_i,x_j)$; set $m_t=(m(x_1),\dots,m(x_t))^\top$.
\State $\Lambda_t(\alpha) \gets K_t + \frac{\sigma^2}{\alpha} I_t$; compute Cholesky $L$ s.t.\ $LL^\top=\Lambda_t(\alpha)$.
\State For any vector $b$, implement $\Lambda_t(\alpha)^{-1}b$ via solves $Lu=b$, $L^{T}v=u$.
\State \Return functions
\[
\mu_{t,\alpha}(x)= m(x) + k_t(x)^{T}\Lambda_t(\alpha)^{-1}\,(y_{1:t}-m_t),\quad
\sigma^2_{t,\alpha}(x)= k_\theta(x,x) - k_t(x)^{T}\Lambda_t(\alpha)^{-1}k_t(x).
\]
\end{algorithmic}
\end{algorithm}

\begin{algorithm}[H]
\caption{\textsc{gEI}$\big(\mu,\widetilde\sigma,\,b,\,g\big)$}
\label{alg:gei}
\begin{algorithmic}[1]
\Require Predictive mean $\mu$, rescaled std.\ dev.\ $\widetilde\sigma>0$, threshold $b$, integer $g\ge 0$.
\State \textbf{If} $\widetilde\sigma=0$ \textbf{then} \Return $0$.
\State $v \gets (b-\mu)/\widetilde\sigma$ \Comment{Standardized gap (nonnegative when $\mu\le \mu^+$)}
\State \textbf{If} $g=0$ \textbf{then} \Return $\Phi(-v)$.
\State \textbf{If} $g=1$ \textbf{then} \Return $\widetilde\sigma\big(\varphi(v)-v\,\Phi(-v)\big)$ \Comment{Classic EI}
\State Compute $\tau_g(v)$ via \textsc{TauSeries}$(v,g)$ (Alg.~\ref{alg:taus}).
\State \Return $\widetilde\sigma^{g}\,\tau_g(v)$.
\end{algorithmic}
\end{algorithm}

\begin{algorithm}[H]
\caption{\textsc{TauSeries}$(v,g)$: compute $\tau_g(v)=\sum_{k=0}^{g}(-1)^k\binom{g}{k} v^k T_{g-k}(v)$}
\label{alg:taus}
\begin{algorithmic}[1]
\Require $v\in[0,\infty)$, integer $g\ge 0$.
\State Define $T_0(z)=\Phi(-z)$, $T_1(z)=\varphi(z)$.
\For{$m=2,3,\dots,g$}
  \State $T_m(z) \gets z^{m-1}\varphi(z) + (m-1)T_{m-2}(z)$ \Comment{Define symbolically for current $z$}
\EndFor
\State $\tau \gets 0$
\For{$k=0$ \textbf{to} $g$}
  \State $\tau \gets \tau + (-1)^k \binom{g}{k} v^k\, T_{g-k}(v)$
\EndFor
\State \Return $\tau$
\end{algorithmic}
\end{algorithm}

\begin{algorithm}[H]
\caption{\textsc{TemperSchedule}$(\mathcal{D}_t)$ \ (examples)}
\label{alg:temper}
\begin{algorithmic}[1]
\Require Data $\mathcal{D}_t$.
\State \textbf{Option A (fixed):} \Return $\alpha_t \gets 1$.
\State \textbf{Option B (information-matching):} See Section \ref{sec:algorithmDesign} and the update in Equation (\ref{eq:hw-alpha-es}).
\State \textbf{Option C (conservatism cap):} Choose $\alpha_t$ as the largest value in $(0,1]$ that satisfies a given conservatism constraint (e.g., on an optimism slope proxy or credible width).
\end{algorithmic}
\end{algorithm}

(i) For numerical stability, compute $\Lambda_t(\alpha)^{-1}\cdot$ via Cholesky solves; add a small jitter if needed.\\
(ii) For integer $g$, the finite expansion for $\tau_g$ uses stable building blocks $T_m$; the recursion $T_m(z)=z^{m-1}\varphi(z)+(m-1)T_{m-2}(z)$ avoids repeated integration.\\
(iii) Choosing $g$ controls exploitation--exploration: $g{=}0$ is PI, $g{=}1$ is classic EI, larger $g$ emphasizes larger improvements.

\section{Proof of GP Regret Bounds} \label{sec:GP=proof}

\subsection{Proof of Proposition \ref{prop:tau-g}}
\begin{proof}
For the integer case $g$, the formula closely aligns with \citep{schonlau1998global}, but is derived for finding the global maximum. Define $u:=\frac{f(x)-\mu_{t,\alpha}(x;\theta)}{\sigma_{t,\alpha}(x;\theta)}$.
For general $g$, note the acquisition function can be reformulated as follows: 
\[
\alpha_{\theta,g}^{EI(f)}=\begin{cases}
\sigma_{t, \alpha}^{g}(x;\theta)(u-v)^{g} & \text{if }u>v\\
0 & \text{otherwise}.
\end{cases}
\]
Rewriting the acquisition above yields
\begin{equation}
\alpha^{\mathrm{EI}(f)}_{\theta,g,\alpha}(x\mid\mathcal D_t)
=\sigma_{t,\alpha}^g(x;\theta)\,\tau_g\!\big(v_{\alpha}\big),
\qquad
\tau_g(v):=\int_v^\infty (u-v)^g\,\phi(u)\,du.
\end{equation}

To show that the function $\tau_g(.)$ has closed-form expression for the integer $g$, observe that the the sole randomness of $\alpha_{\theta,g}^{EI(f)}$ comes from
$u$. Hence we have 
\begin{align*}
\alpha_{\theta,g, \alpha}^{EI(f)} & =\sigma_{t,\alpha}^{g}(x;\theta)\int_{v}^{\infty}\sum_{k=0}^{g}(-1)^{k}\binom{g}{k}u^{g-k}v^{k}\phi(u)du\\
 & =\sigma_{t,\alpha}^{g}(x;\theta)\left\{ \sum_{k=0}^{g}(-1)^{k}\binom{g}{k}v^{k}\int_{v}^{\infty}u^{g-k}\phi(u)du\right\} \\
 & =\sigma_{t,\alpha}^{g}(x;\theta)\left\{ \sum_{k=0}^{g}(-1)^{k}\binom{g}{k}v^{k}T_{g-k}\right\} ,
\end{align*}
where we define $T_{m}:=\int_{v}^{\infty}u^{m}\phi(u)du$. To show
the recursive relationship, note $T_{m}$ can be computed explicitly
through integration by part by writing $u^{m}=u^{m-1}u$. This shows
\[
T_{m}=[-u^{m-1}\phi(u)]_{v}^{\infty}-\int_{v}^{\infty}(-\phi(u))(m-1)u^{m-2}du=v^{m-1}\phi(v)+(m-1)T_{m-2}.
\]
Additionally, it is easy to compute the base cases
\begin{gather*}
T_{0}=\int_{v}^{\infty}\phi(u)du=1-\Phi(v)=\Phi(-v),\\
T_{1}=\int_{v}^{\infty}u\phi(u)du=[-\phi(u)]_{v}^{\infty}=\phi(v).
\end{gather*}
\end{proof}
We study the properties of $\tau_g(.)$ for the integer case in Lemma \ref{lem:tau-g} and the general positive real case in Lemma \ref{lem:tau-g-real}.

\subsection{Analysis of \texorpdfstring{$\tau_{g}$}{taug}}

\begin{lemma}
\label{lem:tau-g} For $g\in\mathcal{N}^{+}$, consider the function
\begin{equation}
\tau_{g}(z):=\sum_{k=0}^{g}(-1)^{k}\binom{g}{k}z^{k}T_{g-k}(z),\label{eq:tau-g-1}
\end{equation}
where $T_{0}(z)=\Phi(-z)$, $T_{1}(z)=\phi(z)$, and $T_{m}(z)=z^{m-1}\phi(z)+(m-1)T_{m-2}(z)$
for $m>1.$ The function $\tau_{g}(.)$ has the following properties: 
\end{lemma}

\begin{enumerate}
\item $T_{m}'(z)=-z^{m}\phi(z).$ 
\item $\tau_{g}(z)$ is a decreasing function in $z$ for any fixed $g\in\mathcal{N}^{+}$. 

\end{enumerate}

\begin{proof}
\begin{enumerate}
\item We can see this by induction. To establish the base case, note that
Since $T_{0}(z)=\Phi(-z)$ and $T_{1}(z)=\phi(z)$, we have $T_{0}'(z)=-\phi(z)$
and $T_{1}'(z)=\phi'(z)=-z\phi(z)$. Suppose this pattern holds for
any integer less than $k$. Since $T_{m}(z)=z^{m-1}\phi(z)+(m-1)T_{m-2}(z)$,
we have 
\begin{align*}
T_{k+1}'(z) & =kz^{k-1}\phi(z)+z^{k}\phi'(z)+kT_{k-1}'(z)\\
 & =kz^{k-1}\phi(z)-z^{k+1}\phi(z)-kz^{k-1}\phi(z)\\
 & =-z^{k+1}\phi(z).
\end{align*}
\item By the product rule, we can decompose the derivative of $\tau_{g}$
into two parts: 
\begin{align*}
\tau_{g}'(z)=\sum_{k=0}^{g}(-1)^{k}\binom{g}{k}kz^{k-1}T_{g-k}(z)+\sum_{k=0}^{g}(-1)^{k}\binom{g}{k}kz^{k}T_{g-k}'(z)
\end{align*}
The first summation can be computed as follows: 
\begin{align*}
\sum_{k=0}^{g}(-1)^{k}\binom{g}{k}kz^{k-1}T_{g-k}(z) & =\sum_{k'=0}^{g-1}(-1)^{k'+1}\binom{g}{k'+1}(k'+1)z^{k'}T_{g-(k'+1)}(z)\\
 & =\sum_{k'=0}^{g-1}(-1)^{k'+1}g\binom{g-1}{k'}z^{k'}T_{g-(k'+1)}(z)\\
 & =-g\sum_{k=0}^{g-1}(-1)^{k'}\binom{g-1}{k'}z^{k'}T_{g-(k'+1)}(z)\\
 & =-g\tau_{g-1}(z).
\end{align*}
It is also possible to show the second summation is zero by the first
part of the lemma: 
\begin{align*}
\sum_{k=0}^{g}(-1)^{k}\binom{g}{k}kz^{k}T_{g-k}'(z) & =\sum_{k=0}^{g}(-1)^{k}\binom{g}{k}kz^{k}(-z^{g-k}\phi(z))\\
 & =-\phi(z)z^{g}\sum_{k=0}^{g}\sum_{k=0}^{g}(-1)^{k}\binom{g}{k}\\
 & =-\phi(z)z^{g}(1-1)^{g}=0.
\end{align*}
Hence we have $\tau_{g}'(z)=-g\tau_{g-1}(z)$. Since the generalized
expected improvement $\alpha_{\theta,g}^{EI(f)}(x|\mathcal{D}_{t})$
is always positive, hence $\tau_{g}'(z)=-g\tau_{g-1}(z)<0$. 

\end{enumerate}
\end{proof}

\begin{lemma} \label{lem:tau-g-real}
For $g \in \mathcal{R}^+$, consider the function
\begin{equation} \label{eq:tau-g-real}
    \tau_g(z) :=\int_z^\infty (u-z)^g\,\phi(u)\,du. 
\end{equation}
The function $\tau_g(.)$ has the following properties:
\begin{enumerate}[label=(\Roman*)]
    \item $\tau_g(z)$ is a decreasing function in $z$ for any fixed $g \geq 0$.
    \item  For $z  \geq  0$, $\tau_g(z) \leq \frac{2^{g/2-1} \Gamma((g+1)/2)}{\sqrt{\pi}}$ for any $g \geq 0$.  
\end{enumerate}

\end{lemma}

\begin{proof}
\begin{enumerate}[label=(\Roman*)]

    \item Differentiate the integral form (Leibniz + chain rule):
\begin{align}
\frac{d}{dz}\tau_g(z) & 
= \frac{d}{dz}\int_{z}^{\infty}(u-z)^g \phi(u)\,du\\
& = \int_{z}^{\infty}\frac{\partial}{\partial z}(u-z)^g \phi(u)\,du\\
& = -g \int_{z}^{\infty}(u-z)^{g-1}\phi(u)\,du\\
& = -g\,\tau_{g-1}(z).
\end{align}

Hence we have
\[
 {\;\tau_g'(z) = -g\,\tau_{g-1}(z)\;} \qquad (g>0).
\]

Note $\tau_0(z) = 1-\Phi(z)$, which is obviously a decreasing function.
    
    \item Since $\tau_g(z)$ is decreasing function in $z$. It is sufficient to show $\tau_g(0) = \frac{2^{g/2-1} \Gamma((g+1)/2)}{\sqrt{\pi}}.$ Following the proof in Proposition \ref{prop:tau-g}, we have
    \begin{align*}
    \tau_g(0) = \int_{0}^{\infty} u^g \phi(u) du &= \frac{1}{\sqrt{2\pi}} \int_{0}^{\infty} u^g e^{-u^2/2} du \\
    & = \frac{1}{\sqrt{2\pi}} \int_{0}^{\infty} (2t)^{\frac{g}{2}} e^{-t} \frac{\sqrt{2}}{2\sqrt{t}} dt \\
    & = \frac{1}{\sqrt{2\pi}} 2^{(g-1)/2} \int_{0}^{\infty} t^{(g-1)/2} e^{-t} dt \\
    &=  \frac{2^{g/2-1} \Gamma((g+1)/2)}{\sqrt{\pi}}.
    \end{align*}

\end{enumerate}
\end{proof}

\subsection{Concentration Result of Tempered Posterior}

We derive the concentration result as illustrated in Proposition  \ref{prop:uniform-alpha-confidence}. The concentration argument by \cite{wang2014theoreticalanalysisbayesianoptimisation} is based on a Hanson-Wright inequality, which is not appropriate in the BO setting, since the design points are chosen
adaptively and the kernel matrix is dependent on the past noise. Our proof is therefore based on a self-normalized concentration argument under predictable sequential
designs. A typical self-normalized concentration argument assumes the hyperparameter $\theta_t$ is fixed. However, in the BO settings, $\theta$ is constantly updated from the past data and is random. We will start by proving a fixed $\theta$ confidence event in section \ref{subsubsec:fixed-concentration}, and then generalize the self-normalization argument to the BO setting in \ref{subsubsec:uniform-concentration}.

\subsubsection{Concentration Result for a fixed \texorpdfstring{$\theta_t$}{theta t}} \label{subsubsec:fixed-concentration}
We define and review the following notations: let $\lambda_{\alpha}:= \frac{\sigma^2}{\alpha}$. For a fixed kernel hyperparameter $\theta$, recall that the $\alpha$-tempered
posterior mean and variance are:
\begin{align*}
    \mu_{t,\alpha}(x;\theta) &= k_t^\theta(x)^\top \left(K_t^\theta+\lambda_\alpha I_t\right)^{-1} y_{1:t}, \\
    \sigma_{t,\alpha}^2(x;\theta) &= k^\theta(x,x) - k_t^\theta(x)^\top \left(K_t^\theta+\lambda_\alpha I_t\right)^{-1} k_t^\theta(x).
\end{align*}
Let $\psi_\theta(x)$ be the feature map of $k^\theta$. The realized tempered information gain along a design $x_{1:t}$ is given by
$$ \mathcal I_{t,\alpha}^{\theta}(x_{1:t}) := \frac12 \log\left|I_t+\alpha\sigma^{-2}K_t^\theta\right|,
$$
so that
$$ \mathcal I_{t,\alpha}^{\theta}(x_{1:t})
\le \gamma_{t,\alpha}^{\theta} :=
\max_{A\subset\mathcal X: |A|=t} \frac12 \log\left|I_t+\alpha\sigma^{-2}K_A^\theta\right|.
$$

\begin{lemma}[Self-normalized noise bound] \label{lem:self-normalized--noise}
Define
\begin{equation} \label{eq:vnsn}
    V_{n,\alpha}^{\theta} := \lambda_\alpha I + \sum_{t=1}^n \psi_\theta(x_t) \psi_\theta(x_t)', \qquad S_n^\theta := \sum_{t=1}^n \varepsilon_t\psi_\theta(x_t).
\end{equation}
Assume that $x_t$ is $\mathcal F_{t-1}$-measurable and that $\varepsilon_t$ is conditionally $\sigma$-subGaussian given $\mathcal F_{t-1}$. Then, for every fixed $n\ge1$ and every
$\delta_n\in(0,1)$, with probability at least $1-\delta_n$,
$$
\left\|S_n^\theta\right\|_{(V_{n,\alpha}^{\theta})^{-1}} \le
\sigma \sqrt{ 2\left( \mathcal I_{n,\alpha}^{\theta}(x_{1:n}) + 1 + \log\frac{1}{\delta_n} \right)
}.
$$
\end{lemma}

\begin{proof}
Let $\widetilde K_n^\theta=\lambda_\alpha^{-1}K_n^\theta$ and define
\begin{align*}
&\widetilde\psi_\theta(x):= \lambda_\alpha^{-1/2}\psi_\theta(x),
\qquad
\widetilde k^\theta(x,x') := \lambda_\alpha^{-1}k^\theta(x,x') = \alpha\sigma^{-2}k^\theta(x,x'), \\
&\widetilde S_n^\theta := \sum_{s=1}^n\varepsilon_s\widetilde\psi_\theta(x_s)
= \lambda_\alpha^{-1/2}S_n^\theta,
\qquad
\widetilde V_n^\theta := I+ \sum_{s=1}^n \widetilde\psi_\theta(x_s)\otimes\widetilde\psi_\theta(x_s) =\lambda_\alpha^{-1}V_{n,\alpha}^{\theta}.
\end{align*}
Observe that $\left\|\widetilde S_n^\theta\right\|_{(\widetilde V_n^\theta)^{-1}}
= \left\|S_n^\theta\right\|_{(V_{n,\alpha}^{\theta})^{-1}}$.

By the finite-dimensional representation of the RKHS norm in Lemma 1 of \citet{Chowdhury17}, we have 
$$
\left\|\widetilde S_n^\theta\right\|_{(\widetilde V_n^\theta)^{-1}}^2 =
\varepsilon_{1:n}^\top \widetilde K_n^\theta \left(\widetilde K_n^\theta+I_n\right)^{-1} \varepsilon_{1:n}.
$$
We further apply the self-normalized concentration inequality in Theorem 1 of \citet{Chowdhury17} to the scaled kernel $\widetilde k^\theta$: set $\eta_n:=\frac{2}{n}$, then with probability at least $1-\delta_n$, we have
$$\varepsilon_{1:n}^\top \left( \left(\widetilde K_n^\theta+\eta_n I_n\right)^{-1}+I_n \right)^{-1} \varepsilon_{1:n}
\le 2\sigma^2 \log \frac{ \sqrt{\left|(1+\eta_n)I_n+\widetilde K_n^\theta\right|}}{
\delta_n}.$$
Given that for every eigenvalue $z\ge0$ and every $\eta_n>0$, $\frac{z}{z+1} \le
\frac{z+\eta_n}{z+\eta_n+1}.$, we see that $\widetilde K_n^\theta
\left(\widetilde K_n^\theta+I_n\right)^{-1} \preceq \left( \left(\widetilde K_n^\theta+\eta_n I_n\right)^{-1}+I_n \right)^{-1}.$ It follows that
$$
\left\|S_n^\theta\right\|_{(V_{n,\alpha}^{\theta})^{-1}}^2 =
\left\|\widetilde S_n^\theta\right\|_{(\widetilde V_n^\theta)^{-1}}^2 \le
2\sigma^2 \log \frac{\sqrt{\left|(1+\eta_n)I_n+\widetilde K_n^\theta\right|}}{
\delta_n}.
$$

The result follows once we simplify the determinant term. Since $\widetilde K_n^\theta =\alpha\sigma^{-2}K_n^\theta$,
we have
\begin{align*}
\log\sqrt{\left|(1+\eta_n)I_n+\widetilde K_n^\theta\right|}
&= \frac12 \log\left|(1+\eta_n)I_n+\widetilde K_n^\theta\right| \\
&= \frac n2\log(1+\eta_n) + \frac12 \log \left|I_n+\frac{1}{1+\eta_n}\widetilde K_n^\theta \right| \\
&\le \frac n2\log\left(1+\frac{2}{n}\right) + \frac12 \log \left|I_n+\widetilde K_n^\theta\right| \\
&\le 1+ \mathcal I_{n,\alpha}^{\theta}(x_{1:n}).
\end{align*}
The last inequality follows since $\frac n2\log\left(1+\frac{2}{n}\right)\le1.$
Therefore,
$$
\left\|S_n^\theta\right\|_{(V_{n,\alpha}^{\theta})^{-1}}^2
\le 2\sigma^2 \left(
\mathcal I_{n,\alpha}^{\theta}(x_{1:n}) + 1 + \log\frac{1}{\delta_n} \right),
$$
and taking square roots gives the claim.
\end{proof}

\begin{theorem}[Posterior concentration for a fixed hyperparameter]
\label{thm:alpha-concentration-fixed}
Fix $\theta$ and assume $k^\theta(x,x)\le 1$ for all $x\in\mathcal X$.
Let $\{\mathcal F_t\}_{t\ge0}$ denote the filtration generated by the BO
history, assume that $x_t$ is $\mathcal F_{t-1}$-measurable, and suppose $y_t=f(x_t)+\varepsilon_t$, where $\varepsilon_t$ is conditionally $\sigma$-subGaussian given
$\mathcal F_{t-1}$. For every fixed $n\geq0$ and every $\eta\in(0,1)$, with probability at least $1-\eta$, simultaneously for all $x\in\mathcal X$,
$$
\left| \mu_{n,\alpha}(x;\theta)-f(x) \right| \le
\varphi_{n,\alpha}^{\theta}(\eta) \sigma_{n,\alpha}(x;\theta),$$
where
$$
\varphi_{n,\alpha}^{\theta}(\eta) :=
\|f\|_{\mathcal H_\theta(\mathcal X)} + \sqrt{2\alpha \left( \gamma_{n,\alpha}^{\theta} + 1 + \log\frac{1}{\eta} \right)}.
$$
\end{theorem}

\begin{proof}
For $n=0$, the claim follows from the RKHS Cauchy--Schwarz inequality, since $\mu_{0,\alpha}=0$ and $|f(x)|\leq\|f\|_{\mathcal H_\theta}\sqrt{k^\theta(x,x)}=\|f\|_{\mathcal H_\theta}\sigma_{0,\alpha}(x;\theta)$. Hence assume $n\geq1$. Recall the definitions of $V_{n,\alpha}^{\theta}$ and $S_n^{\theta}$ as defined in \eqref{eq:vnsn}. The tempered GP posterior mean can be viewed as the kernel ridge regression estimator $\widehat f_{n,\alpha}^{\theta} = \left(V_{n,\alpha}^{\theta}\right)^{-1} \sum_{t=1}^n y_t\psi_\theta(x_t),$ such that $\widehat f_{n,\alpha}^{\theta}(x) = \mu_{n,\alpha}(x;\theta).$

Using $y_s=f(x_s)+\varepsilon_s$, we have the error decomposition

$$
\widehat f_{n,\alpha}^{\theta}-f =  \left(V_{n,\alpha}^{\theta}\right)^{-1} \sum_{t=1}^n y_t\psi_\theta(x_t)-f= \left(V_{n,\alpha}^{\theta}\right)^{-1}S_n^\theta -
\lambda_\alpha \left(V_{n,\alpha}^{\theta}\right)^{-1}f.
$$
Hence, we have

\begin{align*}
\left|\mu_{n,\alpha}(x;\theta)-f(x) \right| &=  \left| \left\langle \psi_\theta(x),
\mu_{n,\alpha}-f \right\rangle_{\mathcal H_\theta} \right| =  \left| \left\langle \psi_\theta(x),
\widehat f_{n,\alpha}^{\theta}-f \right\rangle_{\mathcal H_\theta} \right| \\
& \le \left| \left\langle \psi_\theta(x),  \left(V_{n,\alpha}^{\theta}\right)^{-1}S_n^\theta
 \right\rangle_{\mathcal H_\theta} \right| + \lambda_{\alpha} \left| \left\langle \psi_\theta(x), \left(V_{n,\alpha}^{\theta}\right)^{-1}f
 \right\rangle_{\mathcal H_\theta} \right|    \\
&\le
\left( \|S_n^\theta\|_{(V_{n,\alpha}^{\theta})^{-1}} + \sqrt{\lambda_\alpha}\|f\|_{\mathcal H_\theta(\mathcal X)} \right) \|\psi_\theta(x)\|_{(V_{n,\alpha}^{\theta})^{-1}},
\end{align*}
where the last inequality is due to Cauchy-schwartz, and the fact that $\|f\|_{(V_{n,\alpha}^{\theta})^{-1}} \leq \frac{1}{\sqrt{\lambda_\alpha}} \|f\|_{\mathcal{H}_{\theta}}$.

Since the posterior variance can be represented as the feature-space identity
$\sigma_{n,\alpha}^2(x;\theta) = \lambda_\alpha \|\psi_\theta(x)\|_{(V_{n,\alpha}^{\theta})^{-1}}^2$, we have
$$
\left| \mu_{n,\alpha}(x;\theta)-f(x) \right| \le
\left( \|f\|_{\mathcal H_\theta(\mathcal X)} + \lambda_\alpha^{-1/2} \|S_n^\theta\|_{(V_{n,\alpha}^{\theta})^{-1}} \right) \sigma_{n,\alpha}(x;\theta).
$$
 
By Lemma~\ref{lem:self-normalized--noise}, for any fixed $\delta_n\in(0,1)$, with probability at least $1-\delta_n$, we have 
$$
\left\|S_n^\theta\right\|_{(V_{n,\alpha}^{\theta})^{-1}} \le \sigma
\sqrt{ 2\left( \mathcal I_{n,\alpha}^{\theta}(x_{1:n}) + 1 + \log\frac{1}{\delta_n} \right)}.
$$
Given that $\mathcal I_{n,\alpha}^{\theta}(x_{1:n}) \le \gamma_{n,\alpha}^{\theta}$, and $\lambda_\alpha^{-1/2}\sigma = \sqrt{\alpha}$, we have
$$
\left| \mu_{n,\alpha}(x;\theta)-f(x) \right| \le
\left[ \|f\|_{\mathcal H_\theta(\mathcal X)} + \sqrt{ 2\alpha \left( \gamma_{n,\alpha}^{\theta} + 1 + \log\frac{1}{\delta_n} \right)} \right]
\sigma_{n,\alpha}(x;\theta).
$$

\end{proof}

\subsubsection{Proof of Proposition \ref{prop:uniform-alpha-confidence}} \label{subsubsec:uniform-concentration}

\begin{proof}
Fix a finite horizon $T\ge1$ and $\delta\in(0,1)$. For each $1\le t\le T$, let $N_\Theta(a_t)$ be the covering number of $\Theta$ by balls of radius $a_t$. We will show that, with probability at least $1-\delta$, simultaneously for any $1\le t\le T$, $x\in\mathcal X$, and $\theta\in\Theta$,
$$
\left|\mu_{t-1,\alpha}(x;\theta)-f(x)\right| \le \Phi_{t,\alpha}^{\Theta}(a_t) \sigma_{t-1,\alpha}(x;\theta),
$$
where

\begin{equation} \label{eq:Phi-expression}
\begin{aligned}
  &\Phi_{t,\alpha}^{\Theta}(a_t) := \overline\varphi_{t,\alpha}^{\Theta}(a_t) + \frac{ \operatorname{err}_{t,\alpha}^{\Theta}(a_t)}{s_{t,\alpha}},\\
  & \overline\varphi_{t,\alpha}^{\Theta}(a_t) := B_\Theta + \sqrt{ 2\alpha \left( \overline\gamma_{t-1,\alpha}^{\Theta} + 1 + \log \frac{\pi^2t^2N_\Theta(a_t)}{3\delta} \right)}, \\
  & \operatorname{err}_{t,\alpha}^{\Theta}(a_t) := a_t \left( L_{t-1,\alpha}^{\mu} (\delta/2) + \overline\varphi_{t,\alpha}^{\Theta}(a_t) L_{t-1,\alpha}^{\sigma}(\delta/2) \right), \qquad s_{t,\alpha} := \sqrt{\frac{\sigma^2}{\alpha(t-1)+\sigma^2}}.
\end{aligned}
\end{equation}

For $1\le t\le T$, let $\Theta_t$ be an $a_t$-net of $\Theta$ with cardinality $|\Theta_t|=N_\Theta(a_t)$. With the choice of $\delta_{t,\vartheta} := \frac{3\delta}{\pi^2t^2N_\Theta(a_t)}$, we apply Theorem \ref{thm:alpha-concentration-fixed} with $n=t-1$ and a fixed $\vartheta\in\Theta_t$:
$$
\left| \mu_{t-1,\alpha}(x;\vartheta)-f(x) \right| \le
\left[ \|f\|_{\mathcal H_\vartheta(\mathcal X)} + \sqrt{ 2\alpha \left( \gamma_{t-1,\alpha}^{\vartheta} + 1 + \log\frac{1}{\delta_{t,\vartheta}}\right)} \right] \sigma_{t-1,\alpha}(x;\vartheta)
$$
for all $x\in\mathcal X$.

Since $\|f\|_{\mathcal H_\vartheta(\mathcal X)}\le B_\Theta$, $\gamma_{t-1,\alpha}^{\vartheta}
\le\overline\gamma_{t-1,\alpha}^{\Theta}$, we have 
$$
\left| \mu_{t-1,\alpha}(x;\vartheta)-f(x) \right| \le \overline\varphi_{t,\alpha}^{\Theta}(a_t)\sigma_{t-1,\alpha}(x;\vartheta).
$$
Taking a union bound over all $\vartheta\in\Theta_t$ and $1\le t\le T$, the total failure probability is at most
$$\sum_{t=1}^T N_\Theta(a_t) \frac{3\delta}{\pi^2t^2N_\Theta(a_t)} \le
\sum_{t=1}^\infty \frac{3\delta}{\pi^2t^2} = \frac{\delta}{2}.
$$
It follows that, with probability at least $1-\delta/2$, for any $1\le t\le T$, $\vartheta\in\Theta_t$, and $x\in\mathcal X$, we have
$$
\left|\mu_{t-1,\alpha}(x;\vartheta)-f(x) \right|
\le \overline\varphi_{t,\alpha}^{\Theta}(a_t) \sigma_{t-1,\alpha}(x;\vartheta).
$$
Intersect this event with $\mathcal E_{T,\delta/2}^{\mathrm{Lip}}$ from Assumption \ref{ass:uniform-theta}. The intersection has probability at least $1-\delta$.

Now fix $1\le t\le T$, $x\in\mathcal X$, and an arbitrary $\theta\in\Theta$. Choose $\vartheta_t(\theta)\in\Theta_t$ such that $\|\theta-\vartheta_t(\theta)\|\le a_t$
On the intersection event,
\begin{align*}
\left| \mu_{t-1,\alpha}(x;\theta)-f(x) \right| &\le
\left| \mu_{t-1,\alpha}(x;\theta) - \mu_{t-1,\alpha}(x;\vartheta_t(\theta))
\right| + \left| \mu_{t-1,\alpha}(x;\vartheta_t(\theta)) - f(x) \right| \\
&\le L_{t-1,\alpha}^{\mu}(\delta/2)a_t + \overline\varphi_{t,\alpha}^{\Theta}(a_t) \sigma_{t-1,\alpha}(x;\vartheta_t(\theta)).
\end{align*}
Note that the Lipschitz bound for the posterior standard deviation gives $\sigma_{t-1,\alpha}(x;\vartheta_t(\theta)) \le \sigma_{t-1,\alpha}(x;\theta) + L_{t-1,\alpha}^{\sigma}(\delta/2)a_t$. Hence we have 
\begin{align*}
    \left|\mu_{t-1,\alpha}(x;\theta)-f(x) \right| & \le \overline\varphi_{t,\alpha}^{\Theta}(a_t)
\sigma_{t-1,\alpha}(x;\theta) + a_t \left( L_{t-1,\alpha}^{\mu}(\delta/2) + \overline\varphi_{t,\alpha}^{\Theta}(a_t) L_{t-1,\alpha}^{\sigma}(\delta/2) \right) \\
&= \overline\varphi_{t,\alpha}^{\Theta}(a_t) \sigma_{t-1,\alpha}(x;\theta) + \operatorname{err}_{t,\alpha}^{\Theta}(a_t).
\end{align*}

By the argument in Lemma \ref{lem:mean-concentration} (equation \eqref{eq:sigma-lb}), $\sigma_{t-1,\alpha}(x;\theta) \ge s_{t,\alpha}.$ Hence 
$$
\left| \mu_{t-1,\alpha}(x;\theta)-f(x) \right| \le
\left( \overline\varphi_{t,\alpha}^{\Theta}(a_t) + \frac{\operatorname{err}_{t,\alpha}^{\Theta}(a_t)}{s_{t,\alpha}} \right) \sigma_{t-1,\alpha}(x;\theta).
$$

\end{proof}

\subsection{Auxiliary Lemmas}

\begin{lemma}
\label{lem:mean-concentration} Let $\mu_{\theta_{t},\alpha}^{+}:=\max_{x\in\mathcal{X}}\mu_{t-1,\alpha}(x;\theta_{t})$, then
\begin{equation}
\mu_{\theta,\alpha}^{+}-\mu_{t-1, \alpha}(x_{t})\leq\tau_{g}^{-1}\left(\frac{\Gamma(\frac{g+1}{2})}{2\sqrt{\pi}}\left(\frac{2\sigma^{2}}{\alpha(t-1)+\sigma^{2}}\right)^{g/2} \right)\nu\sigma_{t-1,\alpha}(x_{t};\theta_{t})
\end{equation}
\end{lemma}

\begin{proof}
We define $x_{t}^{+}=\arg\max_{x\in\mathcal{X}}\mu_{t-1, \alpha}(x;\theta_{t})$.
Since $\alpha_{\theta,g, \alpha}^{EI(f)}(x_{t}|\mathcal{D}_{t-1})\geq\alpha_{\theta,g,\alpha}^{EI(f)}(x_t^{+}|\mathcal{D}_{t-1})$,
we have 
\begin{multline*}
\nu^{g}\sigma_{t-1, \alpha}^{g}(x_{t};\theta)
\tau_{g}\!\left(\frac{\mu_{\theta,\alpha}^{+}-\mu_{t-1,\alpha}(x_{t})}
{\nu\sigma_{t-1, \alpha}(x_{t};\theta)}\right)\\
\geq\nu^{g}\sigma_{t-1, \alpha}^{g}(x_{t}^{+};\theta)\tau_g(0)
\end{multline*}
We start by finding a lower bound for the right hand side of the inequality
above. First observe that using the same step as in part (III) of Lemma \ref{lem:tau-g-real}, we have $\tau_g(0) = \frac{2^{g/2-1} \Gamma((g+1)/2)}{\sqrt{\pi}}$. By the same argument as in Lemma 10 of \citet{wang2014theoreticalanalysisbayesianoptimisation},
we have $\sigma_{t-1}^{2}(x_{t}^{+};\theta_{t})\geq\frac{\sigma^{2}}{\alpha(t-1)+\sigma^{2}}$, since 
\begin{align}
\sigma_{t-1,\alpha}^{2}(x_{t}^{+};\theta_{t})
&\geq1-\mathbf{1}^{T}(\mathbf{1}\mathbf{1}^T+\tfrac{\sigma^{2}}{\alpha}\mathbb{I})^{-1}\mathbf{1}\notag\\
&=1-\frac{\|\mathbf{1}\|_{2}^{2}}{\|\mathbf{1}\|_{2}^{2}+\sigma^{2}/\alpha}
=\frac{\sigma^{2}}{\alpha(t-1)+\sigma^{2}}. \label{eq:sigma-lb}
\end{align}

The lower bound above holds since we can bound the posterior variance at $x_t^+$ in the extreme case where all $(t-1)$ observations are collected at the same point. Since $\sigma_{t-1, \alpha}^{g}(x_{t};\theta)\leq1$ is still bounded by the prior variance $k(x,x)$, it follows 
\[
\frac{\Gamma(\frac{g+1}{2})}{2\sqrt{\pi}}\left(\frac{2\sigma^{2}}{\alpha(t-1)+\sigma^{2}}\right)^{g/2}\leq\tau_{g}\left(\frac{\mu_{\theta}^{+}-\mu_{t-1}(x_{t})}{\nu\sigma_{t-1}(x_{t};\theta)}\right).
\]
Since $\tau_{g}(.)$ is a decreasing function by Lemma \ref{lem:tau-g-real} and its inverse is well-defined,
we must have 
\[
\mu_{\theta}^{+}-\mu_{t-1,\alpha}(x_{t})\leq\tau_{g}^{-1}\left( \frac{\Gamma(\frac{g+1}{2})}{2\sqrt{\pi}}\left(\frac{2\sigma^{2}}{\alpha(t-1)+\sigma^{2}}\right)^{g/2}\right)\nu\sigma_{t-1}(x_{t};\theta_{t}).
\]
\end{proof}

\subsection{Proof of Theorem \ref{thm:general-g-alpha}} \label{subsec:positive-g-proof}

\begin{proof}
Define $\mu_{\theta_{t},\alpha}^{+}:=\max_{x\in\mathcal{X}}\mu_{t-1,\alpha}(x;\theta_{t})$,
and $x_{t}:=\arg\max_{x\in\mathcal{X}}\alpha_{\theta,g,\alpha}^{EI(f)}(x|\mathcal{D}_{t-1})$.
We further define the random variable 
\[
I_{t,g,\alpha}^{\theta}(x)=\left(\left[f(x)-\mu_{\theta_t,\alpha}^{+}\right]_+\right )^{g}.
\]
For readability, write $\Phi_t:=\Phi_{t,\alpha}^{\Theta}(a_t)$, 
$\Phi_T^* := \Phi_{T,\alpha}^{\Theta,\mathrm{max}} = \max_{1\le t\le T}\Phi_{t,\alpha}^{\Theta}(a_s)$, and $$q_{t,\alpha,g} := \tau_g^{-1} \left[ C_{(g)} \left(\frac{\sigma^2}{\alpha(t-1)+\sigma^2} \right)^{g/2}\right].$$
We have the instantaneous regret 
\begin{align*}
r_{t} & =f(x^{*})-f(x_t)=[f(x^{*})-\mu_{\theta_{t},\alpha}^{+}]-[f(x_{t})-\mu_{\theta_{t},\alpha}^{+}]\\
 & \leq I_{t,1,\alpha}^{\theta}(x^{*})-[f(x_{t})-\mu_{\theta_{t},\alpha}^{+}].
\end{align*}
We will work on the high-probability event from Proposition \ref{prop:uniform-alpha-confidence}. On this event, simultaneously for all $1\le t\le T$ and all $x\in\mathcal X$,
$$\left| \mu_{t-1,\alpha}(x;\theta_t)-f(x) \right| \le \Phi_t \sigma_{t-1,\alpha}(x;\theta_t).$$
It follows that
\begin{equation} \label{eq:regret-after-uniform-concentration}
r_t \le I_{t,1,\alpha}^{\theta_t}(x^*) + \Phi_t \sigma_{t-1,\alpha}(x_t;\theta_t)
+ \mu_{\theta_t,\alpha}^{+} - \mu_{t-1,\alpha}(x_t;\theta_t).
\end{equation}

To upper bound $I_{t,1,\alpha}^{\theta_t}(x^*)$, we define $z_t(x):=\frac{\mu_{\theta_t,\alpha}^{+}-\mu_{t-1,\alpha}(x;\theta_t)}{\nu_{\theta_t}\sigma_{t-1,\alpha}(x;\theta_t)} \geq 0.$  By Assumption~\ref{ass:nu-scaling},
$\frac{\Phi_t}{\nu_{\theta_t}} \le \frac1{c_\nu}$.  For fixed $g>0$, define $D_{g,\nu}:=\sup_{0\le z\le 1/c_\nu} \frac{1/c_\nu-z}{\tau_g(z)^{1/g}} <\infty$. On the confidence event, we have
\begin{align*}
I_{t,1,\alpha}^{\theta_t}(x^*) &\le [\mu_{t-1,\alpha}(x^*;\theta_t)+\Phi_t\sigma_{t-1,\alpha}(x^*;\theta_t) -\mu_{\theta_t,\alpha}^{+}]_+ \\
&=\nu_{\theta_t} \sigma_{t-1,\alpha}(x^*;\theta_t) \left[\frac{\Phi_t}{\nu_{\theta_t}}-z_t(x^*)\right]_+ \\
&\le \nu_{\theta_t} \sigma_{t-1,\alpha}(x^*;\theta_t) \left[\frac1{c_\nu}-z_t(x^*)\right]_+ \\
&\le D_{g,\nu}\nu_{\theta_t} \sigma_{t-1,\alpha}(x^*;\theta_t) \tau_g(z_t(x^*))^{1/g} \\
&= D_{g,\nu} \left(\alpha_{\theta_t,g,\alpha}^{EI(f)}(x^*\mid\mathcal D_{t-1})\right)^{1/g} \\
& \le D_{g,\nu} \left( \alpha_{\theta_t,g,\alpha}^{EI(f)}(x_t\mid\mathcal D_{t-1})\right)^{1/g} =
D_{g,\nu}\nu_{\theta_t}\sigma_{t-1,\alpha}(x_t;\theta_t) \tau_g(z_t(x_t))^{1/g}.
\end{align*}

Define constant $C_{(g)}:= \frac{2^{g/2} \Gamma((g+1)/2)}{2\sqrt{\pi}}$. By part (III) of Lemma \ref{lem:tau-g-real}, we have
$$I_{t,1,\alpha}^{\theta_t}(x^*) \le D_{g,\nu}C_{(g)}^{1/g} \nu_{\theta_t} \sigma_{t-1,\alpha}(x_t;\theta_t).$$

To bound $(\mu_{\theta_{t},\alpha}^{+}-\mu_{t-1,\alpha}(x_t;\theta_{t}))$,
we apply lemma \ref{lem:mean-concentration},
which yields
\begin{equation} \label{eq:mean-gap-bound-fixed}
\mu_{\theta_t,\alpha}^{+} - \mu_{t-1,\alpha}(x_t;\theta_t) \le
q_{t,\alpha,g} \nu_{\theta_t} \sigma_{t-1,\alpha}(x_t;\theta_t).
\end{equation}
By combining the inequalities above, we can upper bound the simple regret as follows:
\begin{align*} 
r_t &\le \sigma_{t-1,\alpha}(x_t;\theta_t) \Big\{ D_{g,\nu} C_{(g)}^{1/g}\nu_{\theta_t}
+ \Phi_t + q_{t,\alpha,g}\nu_{\theta_t} \Big\}.
\end{align*}
By the argument in Part (II) of Lemma \ref{lem:tau-g-real} and the inverse derivative formula, we may compute the derivative of $\tau_g^{-1}(z)$ as $\frac{1}{\tau_g'(z)} = -\frac{1}{g \tau_{g-1}(z)} \leq 0$. Hence $\tau_g^{(-1)}(.)$ is monotonically decreasing from $(0,c_{0})$, where $c_{0}=\tau_{g}(0) = C_{(g)}$. It follows from the definition of $q_{t,a,g}$ that $q_{t,a,g} \leq q_{T,a,g}$ for $t\leq T$. Furthermore, by the assumed scaling of $\nu_{\theta_t}$, $\nu_{\theta_t} \le C_\nu \Phi_t \le C_\nu \Phi_T^*$. Hence we have
\begin{equation*}
r_t \le \sigma_{t-1,\alpha}(x_t;\theta_t) \Big[ 1 + C_\nu D_{g,\nu} C_{(g)}^{1/g} + C_\nu q_{T,\alpha,g} \Big] \Phi_T^*.
\end{equation*}

Squaring and summing of the instantaneous regret yields:
$$\sum_{t=1}^T r_t^2  \le \Big[ 1 + C_\nu D_{g,\nu} C_{(g)}^{1/g} + C_\nu q_{T,\alpha,g} \Big]^2
(\Phi_T^*)^2 \sum_{t=1}^T \sigma_{t-1,\alpha}^2(x_t;\theta_t).$$
Since $\Theta$ is compact, there exits coordinate wise $(\theta_j^L, \theta_j^U)$ such that $0 < \theta_j^L < \theta_t <\theta_j^U < \infty$. By Lemma 8 of \citet{wang2014theoreticalanalysisbayesianoptimisation},for every $\theta^L\leq \theta_t$,
$$\mathcal H_{\theta}(\mathcal X) \subseteq\mathcal H_{\theta^L}(\mathcal X),
\qquad \|h\|_{\mathcal H_{\theta^L}}^2
\leq c(\theta) \|h\|_{\mathcal H_{\theta}}^2,
\quad c(\theta):=\prod_{j=1}^q\frac{\theta_j}{\theta_j^L}.$$
It follows that $\sigma_{t,\alpha}^2(x;\theta) \leq c(\theta) \sigma_{t,\alpha}^2(x;\theta^L)$. Hence, uniformly over $\Theta$, we have
$$\sigma_{t,\alpha}^2(x_t;\theta_t) \leq C_{\mathrm{var}} \sigma_{t,\alpha}^2(x_t;\theta^L),
\qquad C_{\mathrm{var}} := \prod_{j=1}^q\frac{\theta_j^U}{\theta_j^L}.$$

Since prior variance is normalized to $1$, we have $0\leq \sigma_{t-1,\alpha}^2(x_t; \theta^L)\leq1$. By concavity of $x \mapsto\log(1+\alpha \sigma^{-2} x)$, we also have $$\log(1+\alpha \sigma^{-2} \sigma_{t-1,\alpha}^2(x_t; \theta^L)) \geq \sigma_{t-1,\alpha}^2(x_t; \theta^L)\log(1+\alpha \sigma^{-2}),$$
which implies $$\sigma_{t-1,\alpha}^2(x_t; \theta^L) \leq \frac{\log(1+\alpha\sigma^{-2}\sigma_{t-1,\alpha}^2(x_t; \theta^L))}{\log(1+\alpha\sigma^{-2})}.$$ 
Given that $\frac12
\sum_{t=1}^T\log\!\left(1+\alpha\sigma^{-2}\sigma_{t-1,\alpha}^2(x_t;\theta^0)\right) \leq
\overline\gamma_{T,\alpha}^{\Theta}$, we can control the accumulated posterior variances as
$$\sum_{t=1}^T \sigma_{t-1,\alpha}^2(x_t;\theta_t) \leq \frac{2C_{\mathrm{var}}}{
\log(1+\alpha\sigma^{-2})}\overline\gamma_{T,\alpha}^{\Theta}.$$
It follows that
$$\sum_{t=1}^T r_t^2  \leq \Big[ 1 + C_\nu D_{g,\nu} C_{(g)}^{1/g} + C_\nu q_{T,\alpha,g} \Big]^2
(\Phi_T^*)^2 \left(\frac{2 C_{\mathrm{var}}}{\log(1+\alpha\sigma^{-2})}\overline\gamma_{T,\alpha}^{\Theta} \right).$$
Finally, by $R_{T}^{2}\leq T\sum_{t}r_{t}^{2}$, we have
$$
R_T \le \sqrt{ \frac{ 2C_{\mathrm{var}}T\overline\gamma_{T,\alpha}^{\Theta}}{\log(1+ \alpha \sigma^{-2})}}\,\Big[1+ C_\nu D_{g,\nu} C_{(g)}^{1/g} + C_\nu q_{T,\alpha,g} \Big] \Phi_T^*.
$$

\end{proof}

\subsection{Proof of Corollary \ref{cor:general-g-alpha-order}}

\begin{proof}
    By Theorem \ref{thm:general-g-alpha}, it is sufficient to bound the order of $\Phi_{T,\alpha}^{\Theta,\max}$. First, we recall equation (\ref{eq:Phi-expression}) in the proof of Proposition \ref{prop:uniform-alpha-confidence}, where we have:
    \begin{equation*} 
    \begin{aligned}
      &\Phi_{t,\alpha}^{\Theta}(a_t) := \overline\varphi_{t,\alpha}^{\Theta}(a_t) + \frac{ \operatorname{err}_{t,\alpha}^{\Theta}(a_t)}{s_{t,\alpha}},\\
      & \overline\varphi_{t,\alpha}^{\Theta}(a_t) := B_\Theta + \sqrt{ 2\alpha \left( \overline\gamma_{t-1,\alpha}^{\Theta} + 1 + \log \frac{\pi^2t^2N_\Theta(a_t)}{3\delta} \right)}, \\
      & \operatorname{err}_{t,\alpha}^{\Theta}(a_t) := a_t \left( L_{t-1,\alpha}^{\mu} (\delta/2) + \overline\varphi_{t,\alpha}^{\Theta}(a_t) L_{t-1,\alpha}^{\sigma}(\delta/2) \right), \qquad s_{t,\alpha} := \sqrt{\frac{\sigma^2}{\alpha(t-1)+\sigma^2}}.
    \end{aligned}
    \end{equation*}

We first show that by properly choosing the radius $a_t$, the second term $\frac{ \operatorname{err}_{t,\alpha}^{\Theta}(a_t)}{s_{t,\alpha}}$ can be absorbed into the first. Define $\underline\varphi_{t,\alpha}^{\Theta} := B_\Theta+ \sqrt{2\alpha\left(1+\log\frac{\pi^2t^2}{3\delta}
\right)}.$ Let $a_0>0$ be small enough that the standard compact-set covering bound below holds, and choose
\[
a_t:=\min\!\left\{a_0,\frac{s_{t,\alpha}}{2\left[L_{t-1,\alpha}^{\sigma}(\delta/2)+ L_{t-1,\alpha}^{\mu}(\delta/2)/ \underline\varphi_{t,\alpha}^{\Theta}\right]}\right\}.
\]
Then
$$\frac{\operatorname{err}_{t,\alpha}^{\Theta}(a_t)}{s_{t,\alpha}\overline\varphi_{t,\alpha}^{\Theta}(a_t)
} \le \frac{a_t}{s_{t,\alpha}}\left(L_{t-1,\alpha}^{\sigma}(\delta/2) + \frac{L_{t-1,\alpha}^{\mu}(\delta/2)}{\underline\varphi_{t,\alpha}^{\Theta}}\right)\le \frac12. $$
It follows that $\Phi_{t,\alpha}^{\Theta}(a_t) \le \frac32\overline\varphi_{t,\alpha}^{\Theta}(a_t).$

Since $\Theta\subset\mathbb R^p$ is compact, there exists a constant $C_\Theta<\infty$ such that $N_\Theta(a)\le C_\Theta a^{-p}$ for every $0<a\le a_0$. By the additional polynomial-growth condition in Corollary~\ref{cor:general-g-alpha-order} and the bound $s_{t,\alpha} \geq \frac{\sigma}{\sqrt{T+\sigma^2}}$, the displayed choice satisfies, for constants $c>0$ and $p'>0$ independent of $t$,
$$ a_t\ge cT^{-p'}, \qquad 1\le t\le T.
$$
It follows that $\log N_\Theta(a_t)= \mathcal O(p\log T)$ uniformly for all $1\leq t \leq T$. Since $\overline\gamma_{t-1,\alpha}^{\Theta}\le \overline\gamma_{T,\alpha}^{\Theta}$ and $\log t \leq \log T$, the result follows since
$$\overline\varphi_{t,\alpha}^{\Theta}(a_t) = \mathcal O\!\left( B_\Theta+ \sqrt{ \alpha\left( \overline\gamma_{T,\alpha}^{\Theta} + p\log T + \log(1/\delta)\right)}\right).$$
    
\end{proof}

\subsection{Proof of Theorem \ref{thm:pi-alpha}}

We follow the same notation as in the proof of Theorem \ref{thm:general-g-alpha} in Subsection  \ref{subsec:positive-g-proof}, as the proof structure for the $g=0$ case is very similar to the general $g$ case. 

First observe that in probability of improvement, we must have
\begin{equation} \label{eq:pi-equality}
    \mu_{\theta_t,\alpha}^{+}-\mu_{t-1,\alpha}(x_t;\theta_t)=0.
\end{equation}
To see this, define $z_t(x) := \frac{\mu_{\theta_t,\alpha}^{+}-\mu_{t-1,\alpha}(x;\theta_t)}{\nu_{\theta_t}\sigma_{t-1,\alpha}(x;\theta_t)}.$ Since
$\mu_{\theta_t,\alpha}^{+}=\max_{x\in\mathcal X}\mu_{t-1,\alpha}(x;\theta_t)$, we have $z_t(x)\ge0$ for all $x\in\mathcal X$. For $g=0$,
$$\alpha_{\theta,0,\alpha}^{EI(f)}(x\mid\mathcal D_{t-1})=\tau_0(z_t(x))=\Phi(-z_t(x)).$$
Therefore,
$\alpha_{\theta,0,\alpha}^{EI(f)}(x\mid\mathcal D_{t-1})\le \tau_0(0) = \frac12.$
Let $x_t^+\in \arg\max_{x\in\mathcal X}\mu_{t-1,\alpha}(x;\theta_t)$. Then $z_t(x_t^+)=0$, and hence $\alpha_{\theta,0,\alpha}^{EI(f)}(x_t^+\mid\mathcal D_{t-1})= \frac12.$ By acquisition optimality, we must also have
$\alpha_{\theta,0,\alpha}^{EI(f)}(x_t\mid\mathcal D_{t-1})=\frac12.$ Since $\tau_0$ is strictly decreasing, this proves the claim in \eqref{eq:pi-equality}.

Working on the high-probability event from Proposition~\ref{prop:uniform-alpha-confidence}, and by the same argument of as in the proof of Theorem \ref{thm:general-g-alpha} (see \eqref{eq:regret-after-uniform-concentration}), we have 
$$r_t \le I_{t,1,\alpha}^{\theta_t}(x^*) + \Phi_t \sigma_{t-1,\alpha}(x_t;\theta_t)
+ \mu_{\theta_t,\alpha}^{+} - \mu_{t-1,\alpha}(x_t;\theta_t) = I_{t,1,\alpha}^{\theta_t}(x^*) + \Phi_t \sigma_{t-1,\alpha}(x_t;\theta_t) .$$
By the confidence event,
$$I_{t,1,\alpha}^{\theta_t}(x^*) \le [\mu_{t-1,\alpha}(x^*;\theta_t)+\Phi_t\sigma_{t-1,\alpha}(x^*;\theta_t)
-\mu_{\theta_t,\alpha}^{+}]_+ \le \Phi_t\sigma_{t-1,\alpha}(x^*;\theta_t).$$
It follows that
$$r_t \le \Phi_t \left\{ \sigma_{t-1,\alpha}(x^*;\theta_t) + \sigma_{t-1,\alpha}(x_t;\theta_t) \right\}.$$
Since the kernels are normalized, $\sigma_{t-1,\alpha}(x^*;\theta_t)\le 1.$ By the same argument as in Lemma \ref{lem:mean-concentration} (equation \eqref{eq:sigma-lb}), we have the global posterior variance lower bound $\sigma_{t-1,\alpha}(x_t;\theta_t) \ge s_{t,\alpha} = \sqrt{\frac{\sigma^2}{\alpha(t-1)+\sigma^2}}$. Hence we have $\sigma_{t-1,\alpha}(x^*;\theta_t) \le s_{t,\alpha}^{-1} \sigma_{t-1,\alpha}(x_t;\theta_t)$, which implies
$$r_t \le \Phi_t \left( 1+s_{t,\alpha}^{-1}\right) \sigma_{t-1,\alpha}(x_t;\theta_t).$$
Furthermore since $s_{t,\alpha}^{-1}$ is increasing in $t$ and
$\Phi_t\le\Phi_T^*$, we have
\[
r_t \le \Phi_T^* \left( 1+s_{T,\alpha}^{-1}\right)
\sigma_{t-1,\alpha}(x_t;\theta_t).
\]
Squaring and summing the simple regret yields
\begin{align*}
\sum_{t=1}^T r_t^2
&\le (\Phi_T^*)^2 \left(1+s_{T,\alpha}^{-1}\right)^2
\sum_{t=1}^T \sigma_{t-1,\alpha}^2(x_t;\theta_t)\\
&\le (\Phi_T^*)^2\left(1+s_{T,\alpha}^{-1}\right)^2
\left(\frac{2C_{\mathrm{var}}}{\log(1+\alpha\sigma^{-2})}
\overline\gamma_{T,\alpha}^{\Theta}\right),
\end{align*}
where the last inequality is the same information-gain summation step used in Theorem~\ref{thm:general-g-alpha}. The result follows by Cauchy-schwartz, which yields that
$$R_T \le \sqrt{ \frac{2C_{\mathrm{var}}T\overline\gamma_{T,\alpha}^{\Theta}}{
\log(1+\alpha\sigma^{-2})}}\,\Phi_T^*\left(1+s_{T,\alpha}^{-1}\right).$$
 
\subsection{Regret Monotonicity in \texorpdfstring{$g$}{g}} \label{subsec:g-monotonicity}

Recall that the dependence on $g$ of our regret bound in Theorem \ref{thm:general-g-alpha} is emerged through
$$\kappa_{T,\alpha,g} := 1+D_{g,\nu}C_{\nu}C_{(g)}^{1/g} +C_{\nu}\tau_g^{-1}\!\left[ C_{(g)} \left( \frac{\sigma^2}{\alpha(T-1)+\sigma^2}\right)^{g/2} \right],
\qquad
C_{(g)}:=\frac{2^{g/2}\Gamma((g+1)/2)}{2\sqrt{\pi}}.
$$
The following proposition makes such dependence more transparent:

\begin{proposition}[Dependence on the improvement order] \label{prop:g-monotonicity}
Fix $c_\nu>0$, $\alpha\in(0,1]$, and $T>1$. Recall that $C_{(g)}=\tau_g(0) = \frac{2^{g/2}\Gamma((g+1)/2)}{2\sqrt{\pi}}$, and $$D_{g,\nu} = \sup_{0\leq z\leq1/c_\nu} \frac{1/c_\nu-z}{\tau_g(z)^{1/g}}.$$ Then
$D_{g,\nu}C_{(g)}^{1/g}$ is nonincreasing in $g>0$ and diverges as $g\downarrow0$. In contrast, $$\tau_g^{-1}\!\left[ C_{(g)} \left( \frac{\sigma^2} {\alpha(T-1)+\sigma^2} \right)^{g/2}\right]$$
is nondecreasing in $g>0$ and diverges as $g\to\infty$.
\end{proposition}
\begin{proof}
    For $g>0$, let $U_g$ be a random variable with density $p_g(u)=\frac{u^g\phi(u)}{C_{(g)}}$, $u>0$. Then $U_g^2\sim\chi^2_{g+1}$ and, for every $z\geq0$,
$$\left\{\frac{\tau_g(z)}{C_{(g)}}\right\}^{1/g}=\left\|\left(1-\frac z{U_g}\right)_+\right\|_{L^g}.$$ For any $0<g_1<g_2$, we have $U_{g_2} \geq U_{g_1}$
\begin{align*}
\left\{\frac{\tau_{g_2}(z)}{C_{(g_2)}}\right\}^{1/g_2}
&= \left\|\left(1-\frac z{U_{g_2}}\right)_+\right\|_{L^{g_2}}\\
&\ge \left\| \left(1-\frac z{U_{g_1}}\right)_+ \right\|_{L^{g_2}}
\geq \left\|\left(1-\frac z{U_{g_1}}\right)_+ \right\|_{L^{g_1}}\\
&= \left\{ \frac{\tau_{g_1}(z)}{C_{(g_1)}} \right\}^{1/g_1}.
\end{align*}
Hence we have $\left\{ \frac{\tau_g(z)}{C_{(g)}} \right\}^{1/g}$ is nondecreasing in $g$ for every fixed $z\geq0$. To prove divergence of $D_{g,\nu}C_{(g)}^{1/g}$ as $g\downarrow0$, fix any $z_0\in(0,1/c_\nu)$, and we have $$\frac{\tau_g(z_0)}{C_{(g)}} \rightarrow \frac{\Phi(-z_0)}{1/2} = 2\Phi(-z_0)<1.$$
It follows that $D_{g,\nu}C_{(g)}^{1/g} \geq \frac{1/c_\nu-z_0} {\{\tau_g(z_0)/C_{(g)}\}^{1/g}} \longrightarrow\infty.$

To show divergence of the second claim: Use the integral form $\tau_g(z)=\int_{0}^{\infty} t^{g}\phi(z+t)\,dt=\phi(z)\int_{0}^{\infty} t^{g}e^{-zt-\tfrac12 t^{2}}dt$ and a Laplace approximation as $z\to\infty$, we obtain $\tau_g(z)=\dfrac{\Gamma(g+1)}{\sqrt{2\pi}}\dfrac{e^{-z^{2}/2}}{z^{\,g+1}}\bigl(1+O(z^{-2})\bigr)$.  Let $y_T=C_{(g)} r^{g/2}$ and $r=\dfrac{\sigma^{2}}{\alpha(T-1)+\sigma^{2}}$, and set $-\log y_T=-\log C_{(g)}-\tfrac{g}{2}\log r = -\log C_{(g)}+\tfrac{g}{2}\log\!\bigl(\tfrac{\alpha(T-1)+\sigma^{2}}{\sigma^{2}}\bigr)$.  Inverting the asymptotic yields \[
\tau_g^{-1}\!\bigl(C_{(g)} r^{g/2}\bigr)=\Theta\!\Big(\sqrt{\,g\log\!\bigl(\tfrac{\alpha(T-1)+\sigma^{2}}{\sigma^{2}}\bigr)}\Big)=\Theta\!\big(\sqrt{g\log T}\big).
\]
\end{proof}

\section{Proof of Proposition \ref{prop:calibration2}}
\begin{proof}
Define the residual
\[
r_s:=y_s-\mu_{s-1,1}(x_s)=\varepsilon_s-e_s.
\]
Then
\[
\mathrm{MSE}_t
=\frac1t\sum_{s=1}^t\varepsilon_s^2
 +\frac1t\sum_{s=1}^t e_s^2
 -\frac2t\sum_{s=1}^t\varepsilon_s e_s.
\]
Let $d_s:=\varepsilon_s^2-\sigma^2$. The sequence $\{d_s\}$ is a martingale-difference sequence and, by the conditional fourth-moment assumption,
\[
\mathbb E[d_s^2\mid\mathcal F_{s-1}]
\leq 2M_4+2\sigma^4.
\]
Martingale orthogonality therefore gives
\[
\mathbb E\!\left[\left(\frac1t\sum_{s=1}^t d_s\right)^2\right]
\leq \frac{2M_4+2\sigma^4}{t}\longrightarrow0,
\]
so $t^{-1}\sum_{s=1}^t\varepsilon_s^2\xrightarrow{p}\sigma^2$.
Likewise, $e_s$ is $\mathcal F_{s-1}$-measurable and bounded by $E$, so $\{\varepsilon_s e_s\}$ is a martingale-difference sequence with conditional second moment at most $E^2\sigma^2$. Hence
\[
\frac1t\sum_{s=1}^t\varepsilon_s e_s\xrightarrow{p}0.
\]
Consequently,
\[
\mathrm{MSE}_t
=\sigma^2+\frac1t\sum_{s=1}^t e_s^2+o_p(1).
\]
In case (i), $\mathrm{MSE}_t\xrightarrow{p}\sigma^2$; in case (ii), $\mathrm{MSE}_t\xrightarrow{p}\sigma^2+b^2$. Since $\widehat\sigma_t^2\xrightarrow{p}\sigma^2$ and $\mathrm{PV}_t\xrightarrow{p}\mathrm{PV}_\infty$, the claimed limits follow from the continuous-mapping theorem applied to
\[
\widehat\alpha_t
=\min\!\left\{\sqrt{\frac{\mathrm{PV}_t+\widehat\sigma_t^2}
{\mathrm{PV}_t+\mathrm{MSE}_t}},1\right\}.
\]
In case (ii), the limiting ratio is strictly below one, so the truncation is asymptotically inactive.
\end{proof}

\section{Additional Simulation Results} \label{subsec:appendix-simulation}

\subsection{Lower Noise Setting} \label{subsec:low-noise}

We repeat the same experiment as illustrated in Section \ref{subsec:simulation} but with lower $\sigma=0.01$ Gaussian noise settings. All configurations of the $61$ functions, choices of $g\in \{0, 1, 2\}$, and choices of $\alpha \in \{\text{tempered}, 1\}$ remain the same. The exact functional forms can be found in \citep{simulationlib}. Given this is the very low noise setting, we focus on best observed rather than regret as focused in the main text.

The main empirical finding in this lower noise setting can be found in Table \ref{tab:alpha_by_g}. In this lower noise setting, tempering the posterior can substantially improve Bayesian optimization performance in less exploratory settings, particularly when $g$ is small and the acquisition function tends to overexploit. 

For $g=0$, our adaptive tempering approach yields a substantial and statistically significant improvement ($p=0.041$): the tempered posterior wins $38$ out of $61$ cases, achieves a lower average rank, and exhibits much smaller raw and absolute margins. This behavior not only echos our discussion of Theorem \ref{thm:general-g-alpha}, but also aligns with our theoretical intuition: as probability of improvement is known to be overly exploitative \citep{garnett_bayesoptbook_2023}, tempering the posterior effectively encourages additional principled exploration and tends to yields better optimization results.

Tempering also offers a modest advantage for the $g=1$ expected improvement setting, leading to higher win rates and better average ranks. Although the effect is only marginally significant in our experimentation ($p=0.097$), both the raw and normalized margin metrics are more stable under the tempered posterior. This indicates that in instances where the regular posterior outperforms tempering, its advantage tends to be relatively small compared to the cases where tempering wins.

Notably, tempering the posterior significantly deteriorates the situation where $g=2$ based on Panel C of Table \ref{tab:alpha_by_g}. A plausible explanation is that $g=2$ already promotes aggressive exploration, and tempering further inflates posterior uncertainty, inducing excessive exploration and thus degrading finite-sample performance. The overly exploration behavior can also lead to larger MSE, which induces stronger tempering effect of $\alpha$ based on the estimator in equation (\ref{eq:hw-alpha-es}). 

\begin{table}[t]
\centering
\small
\setlength{\tabcolsep}{4pt}
\renewcommand{\arraystretch}{1.1}
\caption{Paired comparison of $\alpha\in\{1,\mathrm{hw}\}$ across 61 functions (seed-averaged) for each $g$.
Lower \emph{Avg. rank} is better; \emph{Avg. margin to best} closer to $0$ (less negative) is better.
$p$-values from Wilcoxon paired signed-rank (one-sided, $H_1$: the median of the paired difference is larger than $0$ (tempereing is better)).}
\label{tab:alpha_by_g}
\begin{threeparttable}
\begin{tabular}{l
  S[table-format=2.0, round-mode=places, round-precision=0] 
  r                                                         
  S[table-format=1.3]                                       
  S[table-format=-3.3]                                      
  S[table-format=-1.3]                                      
}
\toprule
\multicolumn{1}{c}{$\alpha$} &
\multicolumn{1}{c}{Wins} &
\multicolumn{1}{c}{Strict win rate} &
\multicolumn{1}{c}{Avg. rank $\downarrow$} &
\multicolumn{1}{c}{Avg. margin to best $\uparrow$} &
\multicolumn{1}{c}{Avg. norm. margin $\uparrow$} \\
\midrule
\multicolumn{6}{l}{\textsc{Panel A:} $g=0$ \hfill  $p=0.041$} \\
\cmidrule(lr){1-6}
1.0      & 23 & 37.7\% & 1.631 & -30.864 & -0.6333 \\
Tempered & 38 & 62.3\% & 1.369 &  -2.946 & -0.3667 \\
\addlinespace[4pt]
\multicolumn{6}{l}{\textsc{Panel B:} $g=1$ \hfill  $p=0.097$} \\
\cmidrule(lr){1-6}
1.0        & 28 & 45.9\% & 1.55 & -24.25 & -0.541 \\
Tempered     & 33 & 54.1\% & 1.46 & -11.18 & -0.443 \\
\addlinespace[4pt]
\multicolumn{6}{l}{\textsc{Panel C:} $g=2$ \hfill  $p=0.932$} \\
\cmidrule(lr){1-6}
1.0        & 39 & 63.9\% & 1.37 & -6.33 & -0.361 \\
Tempered       & 22 & 36.1\% & 1.63 & -12.12 & -0.623 \\
\bottomrule
\end{tabular}
\begin{tablenotes}[flushleft]
\footnotesize
\item Notes: Avg margin to best is $\frac{1}{|F|}\sum_{f \in F} \big(y_{f,a}-\max_{a'} y_{f,a'}\big)$ (thus $\le 0$ and larger is better).
Avg normalized margin (range) uses range normalization: $\frac{1}{|F|}\sum_f \frac{y_{f,a}-\max_{a'} y_{f,a'}}{\max_{a'} y_{f,a'} - \min_{a'} y_{f,a'} + \varepsilon}\in[-1,0]$.
\end{tablenotes}
\end{threeparttable}
\end{table}

\subsection{Detailed Table for high Noise Setting} \label{subsec:dth}

\begin{table*}[t]
\centering
\small
\setlength{\tabcolsep}{6pt}
\renewcommand{\arraystretch}{1.08}
\caption{Paired comparison of tempering schedule against the untempered baseline $\alpha=1$}
\label{tab:alpha_noise_sigma2-large}
\begin{threeparttable}
\begin{tabular}{
c
l
S[table-format=2.0, round-mode=places, round-precision=0]
r
S[table-format=1.3]
S[table-format=-1.3]
}
\toprule
\multicolumn{1}{c}{$g$} &
\multicolumn{1}{c}{$\alpha$} &
\multicolumn{1}{c}{Strict Wins} &
\multicolumn{1}{c}{Strict win rate} &
\multicolumn{1}{c}{Avg. rank $\downarrow$} &
\multicolumn{1}{c}{Avg. norm. margin $\uparrow$} \\
\midrule

\multicolumn{6}{l}{\textsc{Panel A: Gaussian noise}} \\
\cmidrule(lr){1-6}
\multirow{2}{*}{$0$}
  & $1.0$      & 25 & 41.0\% & 1.549 & -0.508 \\
  & Tempered   & 31 & 50.8\% & 1.451 & -0.410 \\
\addlinespace[2pt]
\multirow{2}{*}{$1$}
  & $1.0$      & 17 & 27.9\% & 1.680 & -0.639 \\
  & Tempered   & 39 & 63.9\% & 1.320 & -0.279 \\
\addlinespace[2pt]
\multirow{2}{*}{$2$}
  & $1.0$      & 28 & 45.9\% & 1.500 & -0.459 \\
  & Tempered   & 28 & 45.9\% & 1.500 & -0.459 \\

\addlinespace[5pt]
\multicolumn{6}{l}{\textsc{Panel B: Gaussian-mixture noise}} \\
\cmidrule(lr){1-6}
\multirow{2}{*}{$0$}
  & $1.0$      & 25 & 41.0\% & 1.582 & -0.574 \\
  & Tempered   & 35 & 57.4\% & 1.418 & -0.410 \\
\addlinespace[2pt]
\multirow{2}{*}{$1$}
  & $1.0$      & 22 & 36.1\% & 1.623 & -0.607 \\
  & Tempered   & 37 & 60.7\% & 1.377 & -0.361 \\
\addlinespace[2pt]
\multirow{2}{*}{$2$}
  & $1.0$      & 24 & 39.3\% & 1.582 & -0.557 \\
  & Tempered   & 34 & 55.7\% & 1.418 & -0.393 \\

\addlinespace[5pt]
\multicolumn{6}{l}{\textsc{Panel C: Heteroscedastic noise}} \\
\cmidrule(lr){1-6}
\multirow{2}{*}{$0$}
  & $1.0$      & 33 & 54.1\% & 1.434 & -0.410 \\
  & Tempered   & 25 & 41.0\% & 1.566 & -0.541 \\
\addlinespace[2pt]
\multirow{2}{*}{$1$}
  & $1.0$      & 18 & 29.5\% & 1.648 & -0.590 \\
  & Tempered   & 36 & 59.0\% & 1.352 & -0.295 \\
\addlinespace[2pt]
\multirow{2}{*}{$2$}
  & $1.0$      & 26 & 42.6\% & 1.541 & -0.508 \\
  & Tempered   & 31 & 50.8\% & 1.459 & -0.426 \\

\addlinespace[5pt]
\multicolumn{6}{l}{\textsc{Panel D: Student-$t$ noise}} \\
\cmidrule(lr){1-6}
\multirow{2}{*}{$0$}
  & $1.0$      & 26 & 42.6\% & 1.541 & -0.508 \\
  & Tempered   & 31 & 50.8\% & 1.459 & -0.426 \\
\addlinespace[2pt]
\multirow{2}{*}{$1$}
  & $1.0$      & 31 & 50.8\% & 1.475 & -0.459 \\
  & Tempered   & 28 & 45.9\% & 1.525 & -0.508 \\
\addlinespace[2pt]
\multirow{2}{*}{$2$}
  & $1.0$      & 27 & 44.3\% & 1.533 & -0.508 \\
  & Tempered   & 31 & 50.8\% & 1.467 & -0.443 \\

\bottomrule
\end{tabular}

\begin{tablenotes}[flushleft]
\footnotesize
\item Notes: Strict win rates needs not to sum to $100\%$ because there exist exact ties.
Lower Avg. rank is better. Avg. normalized margin is computed as
$\frac{\min_{a'} R_{f,a'} - R_{f,a}} {\max_{a'} R_{f,a'}-\min_{a'} R_{f,a'}+\varepsilon}$, so values are non-positive and larger values are better.
\end{tablenotes}
\end{threeparttable}
\end{table*}

\end{document}